\documentclass{birkjour}
\pdfoutput=1
\usepackage[ascii]{inputenc}
\usepackage{hyperref,dcolumn,bbm,url,booktabs,braket,microtype,aliascnt,mathtools,mathrsfs,etoolbox,cancel,capt-of,algpseudocode,float,newfloat,amsmath,amssymb,color,amsthm,cleveref,tikz}
\usepackage[T1]{fontenc}
\usepackage[USenglish]{babel}

\newtheorem{thm}{Theorem}\crefname{thm}{theorem}{theorems}
\newtheorem{lem}[thm]{Lemma}\crefname{lem}{lemma}{lemmas}
\newtheorem{cor}[thm]{Corollary}\crefname{cor}{corollary}{corollaries}
\newtheorem{prp}[thm]{Proposition}\crefname{prp}{proposition}{propositions}
\newtheorem{prb}[thm]{Problem}\crefname{prb}{problem}{problems}
\newtheorem{dfn}[thm]{Definition}\crefname{dfn}{definition}{definitions}

\DeclareMathOperator{\SU}{SU}
\DeclareMathOperator{\U}{U}

\DeclareMathOperator{\tr}{tr}
\DeclareMathOperator{\spec}{spec}

\DeclareMathOperator{\Ha}{H}

\DeclareMathOperator{\poly}{poly}
\DeclareMathOperator{\Hom}{Hom}
\DeclareMathOperator{\Ind}{Ind}
\newcommand{\CG}{\Phi}
\newcommand{\calH}{\mathcal H}
\newcommand{\calK}{\mathcal K}
\newcommand{\calL}{\mathcal L}
\newcommand{\calA}{\mathcal A}
\newcommand{\calB}{\mathcal B}
\newcommand{\calC}{\mathcal C}

\newcommand{\proj}[1]{\lvert#1\rangle\langle#1\rvert}
\newcommand{\abs}[1]{\lvert#1\rvert}
\newcommand{\norm}[1]{\lVert#1\rVert}
\newcommand{\normHS}[1]{\norm{#1}_\mathrm{HS}}
\newcommand{\id}{\mathbbm 1}
\usetikzlibrary{matrix,arrows,external,shapes,positioning,decorations.markings,calc}
\newcommand*\circled[1]{\tikz[baseline=(char.base)]{\node[shape=circle,draw,inner sep=1pt] (char) {#1};}}
\def\centerarc[#1](#2)(#3:#4:#5){\draw[#1]($(#2)+({#5*cos(#3)},{#5*sin(#3)})$) arc (#3:#4:#5);}
\tikzstyle epr=[fill=black]
\tikzstyle adjoint epr=[fill=black]
\tikzstyle adjoint=[]
\tikzstyle arrowstyle=[scale=1.25]
\tikzset{->-/.style={decoration={markings,mark=at position #1 with {\arrow[arrowstyle]{stealth};}},postaction={decorate}}}
\tikzset{-<-/.style={decoration={markings,mark=at position #1 with {\arrowreversed[arrowstyle]{stealth};}},postaction={decorate}}}

\newcommand{\strongsixjCG}[4]{
  \begin{tikzpicture}[baseline=0cm]
    \draw[-<-=.3] (0,-0.2) -- (0,-0.7);
    \draw[->-=.6] (0.177,0.1) -- (0.766, 0.5);
    \draw[->-=.6] (-0.177,0.1) -- (-0.766,0.5);
    \draw (0,0) circle[radius=0.2];
    \draw (0,0) node {#4};
    \draw (-0.766,0.5) node[above] {#2};
    \draw (0.766, 0.5) node[above] {#3};
    \draw (0,-0.7) node[below] {#1};
  \end{tikzpicture}
}

\newcommand{\strongsixjCGadjoint}[4]{
  \begin{tikzpicture}[baseline=0cm]
    \draw[->-=.6] (0,-0.2) -- (0,-0.7);
    \draw[-<-=.35] (0.177,0.1) -- (0.766, 0.5);
    \draw[-<-=.35] (-0.177,0.1) -- (-0.766,0.5);
    \draw[adjoint] (0,0) circle[radius=0.2];
    \draw (0,0) node {#4};
    \draw (-0.766,0.5) node[above] {#2};
    \draw (0.766, 0.5) node[above] {#3};
    \draw (0,-0.7) node[below] {#1};
  \end{tikzpicture}
}

\newcommand{\strongsixjCGinvariant}[1]{
    \begin{tikzpicture}[baseline=0cm]
      \draw[->-=.8] (0,-0.45) -- (0,-0.2);
      \draw[->-=.65] (0,-0.6) -- (0,-0.9);
      \draw[->-=.6] (0.177,0.1) -- (0.766, 0.5);
      \draw[->-=.6] (-0.177,0.1) -- (-0.766,0.5);
      \draw (0,0) circle[radius=0.2];
      \draw (0,0) node {#1};
      \draw[epr] (0,-0.52) circle[radius=0.1];
      \draw (-0.766,0.5) node[above] {$\alpha$};
      \draw (0.766, 0.5) node[above] {$\beta$};
      \draw (0,-0.8) node[below] {$\lambda$};
    \end{tikzpicture}
}

\newcommand{\strongsixjCGadjointinvariant}[1]{
    \begin{tikzpicture}[baseline=0cm]
      \draw[->-=.6] (0,-0.2) -- (0,-0.7);
      \draw[->-=.7] (0.57,0.35) -- (0.786,0.5);
      \draw[->-=.7] (0.413, 0.25) -- (0.177,0.1);
      \draw[->-=.7] (-0.57,0.35) -- (-0.786,0.5);
      \draw[->-=.7] (-0.413, 0.25) -- (-0.177,0.1);
      \draw[adjoint] (0,0) circle[radius= 0.2];
      \draw[epr] (0.491, 0.3) circle[radius=0.1];
      \draw[epr] (-0.491, 0.3) circle[radius=0.1];
      \draw (0,0) node {#1};
      \draw (-0.766,0.5) node[above] {$\alpha$};
      \draw (0.766, 0.5) node[above] {$\beta$};
      \draw (0,-0.7) node[below] {$\lambda$};
    \end{tikzpicture}
}

\newcommand{\strongsixjAsymmetric}{
  \begin{tikzpicture}[baseline=0cm]
    \draw[->-=.6] (0.173,0.1) -- (0.693, 0.4);
    \draw[->-=.6] (-0.173, 0.1) -- (-0.693, 0.4);
    \draw[-<-=.35] (0,-0.2) -- (0,-0.8);
    \draw[adjoint] (0.866, 0.5) circle[radius= 0.2];
    \draw[adjoint] (-0.866, 0.5) circle[radius= 0.2];
    \draw (0, -1) circle[radius= 0.2];
    \draw (0,0) circle[radius= 0.2];
    \draw[->-=.6] (0.19,-0.98) arc [radius=1, start angle= -77, end angle= 18];
    \draw[->-=.6] (0.748,0.662) arc [radius=1, start angle= 43, end angle= 138];
    \draw[->-=.6] (-0.967,0.319) arc [radius=1, start angle= 163, end angle= 259];
    \draw (0,-1) node {$i$};
    \draw (0,0) node {$j$};
    \draw (0.866, 0.5) node {$k$};
    \draw (-0.866, 0.5) node {$l$};
    \draw (-0.53,0.3) node[below] {$\alpha$};
    \draw (0.53,0.3) node[below] {$\beta$};
    \draw (0.866,-0.5) node[below] {$\gamma$};
    \draw (0, -0.5) node[right] {$\mu$};
    \draw (0,1) node[below] {$\nu$};
    \draw (-0.866,-0.5) node[below] {$\lambda$};
  \end{tikzpicture}
}

\newcommand{\strongsixjAsymmetricTeleported}{
  \begin{tikzpicture}[baseline=0cm, scale=1.4]
    \draw (-0.53,0.2) node[below] {$\alpha$};
    \draw[->-=.55] (-0.54,0.31) -- (-0.866, 0.5);
    \draw[->-=.66] (-0.54,0.31) -- (-0.33, 0.18);
    \draw[->-=.7] (0,0) -- (-0.33, 0.18);
    \draw[adjoint epr] (-0.56,0.32) circle[radius=0.1/1.4];
    \draw[adjoint epr] (-0.31,0.17) circle[radius=0.1/1.4];

    \draw (0.53,0.3) node[below] {$\beta$};
    \draw[->-=.55] (0.54,0.31) -- (0.866, 0.5);
    \draw[->-=.66] (0.54,0.31) -- (0.33, 0.18);
    \draw[->-=.7] (0,0) -- (0.33, 0.18);
    \draw[adjoint epr] (0.56,0.32) circle[radius=0.1/1.4];
    \draw[adjoint epr] (0.31,0.17) circle[radius=0.1/1.4];

    \draw (0, -0.5) node[right] {$\mu$};
    \draw[->-=.75] (0,-1) -- (0,-0.65);
    \draw[->-=.65] (0,-0.35) -- (0,-0.65);
    \draw[->-=.5] (0,-0.35) -- (0,0);
    \draw[adjoint epr] (0,-0.355) circle[radius=0.1/1.4];
    \draw[adjoint epr] (0,-0.65) circle[radius=0.1/1.4];

    \draw (0,1) node[below] {$\nu$};
    \centerarc[->-=0.57](0,0)(35:70:1);
    \centerarc[-<-=0.35](0,0)(70:120:1);
    \centerarc[->-=0.5](0,0)(120:145:1);
    \draw[adjoint epr] (0.35,0.92) circle[radius=0.1/1.4];
    \draw[adjoint epr] (-0.35,0.92) circle[radius=0.1/1.4];

    \draw (0.866,-0.5) node[right] {$\gamma$};
    \centerarc[->-=0.75](0,0)(-90:-50:1);
    \centerarc[-<-=0.35](0,0)(-40:-13:1);
    \centerarc[->-=0.65](0,0)(-10:25:1);
    \draw[adjoint epr] (1,-0.15) circle[radius=0.1/1.4];
    \draw[adjoint epr] (0.7,-0.7) circle[radius=0.1/1.4];

    \draw (-0.866,-0.5) node[left] {$\lambda$};
    \centerarc[-<-=0.6](0,0)(270:230:1);
    \centerarc[->-=0.6](0,0)(220:193:1);
    \centerarc[-<-=0.4](0,0)(190:150:1);
    \draw[adjoint epr] (-1,-0.15) circle[radius=0.1/1.4];
    \draw[adjoint epr] (-0.7,-0.7) circle[radius=0.1/1.4];

    \draw[fill=white] (0,-1) circle[radius=0.2/1.4];
    \draw (0,-1) node {$i$};
    \draw[fill=white] (0,0) circle[radius=0.2/1.4];
    \draw (0,0) node {$j$};
    \draw[fill=white] (0.866, 0.5) circle[radius=0.2/1.4];
    \draw (0.866, 0.5) node {$k$};
    \draw[fill=white] (-0.866, 0.5) circle[radius=0.2/1.4];
    \draw (-0.866, 0.5) node {$l$};
  \end{tikzpicture}
}

\newcommand{\strongsixjSymmetric}{
  \begin{tikzpicture}[baseline=0cm]
    \draw (0,-1) node {$i$};
    \draw (0,-1) circle[radius= 0.2];
    \draw (0,0) node {$j$};
    \draw (0,0) circle[radius= 0.2];
    \draw (0.866, 0.5) node {$k$};
    \draw (0.866, 0.5) circle[radius= 0.2];
    \draw (-0.866, 0.5) node {$l$};
    \draw (-0.866, 0.5) circle[radius= 0.2];

    \draw (-0.53,0.2) node[below] {$\alpha$};
    \draw[->-=.56] (-0.693,0.4) -- (-0.433, 0.25);
    \draw[->-=.56] (-0.173,0.1) -- (-0.433, 0.25);
    \draw[adjoint epr] (-0.43,0.24) circle[radius=0.1];

    \draw (0.53,0.3) node[below] {$\beta$};
    \draw[->-=.56] (0.693,0.4) -- (0.433, 0.25);
    \draw[->-=.56] (0.173,0.1) -- (0.433, 0.25);
    \draw[adjoint epr] (0.43,0.24) circle[radius=0.1];

    \draw (0, -0.5) node[right] {$\mu$};
    \draw[->-=.56] (0,-0.2) -- (0,-0.5);
    \draw[->-=.56] (0,-0.8) -- (0,-0.48);
    \draw[adjoint epr] (0,-0.5) circle[radius=0.1];

    \draw (0,1) node[below] {$\nu$};
    \draw[->-=.56] (0.748,0.662) arc [radius=1, start angle=43, end angle=90]; 
    \draw[-<-=.56] (0,1) arc [radius=1, start angle=90, end angle=138];
    \draw[adjoint epr] (0,0.99) circle[radius=0.1];

    \draw (0.866,-0.5) node[below] {$\gamma$};
    \draw[->-=.56] (0.19,-0.98) arc [radius=1, start angle=-77, end angle=-30]; 
    \draw[-<-=.56] (0.83,-0.5) arc [radius=1, start angle=-30, end angle=18];
    \draw[adjoint epr] (0.86,-0.45) circle[radius=0.1];

    \draw (-0.866,-0.5) node[below] {$\lambda$};
    \draw[->-=.56] (-0.967,0.319) arc [radius=1, start angle=163, end angle=210]; 
    \draw[-<-=.56] (-0.866,-0.5) arc [radius=1, start angle=210, end angle=259];
    \draw[adjoint epr] (-0.87,-0.45) circle[radius=0.1];
  \end{tikzpicture}
}

\newcommand{\sixj}[6]{\begin{bmatrix}#1 & #2 & #4\\#3 & #6 & #5\end{bmatrix}}

\allowdisplaybreaks[4]

\newcommand*{\CC}{\mathbb C}
\newcommand*{\RR}{\mathbb R}

\begin{document}

\title{Recoupling coefficients and quantum entropies}
\author{Matthias Christandl}
\address{Institute for Theoretical Physics, ETH Zurich, Wolfgang--Pauli--Strasse 27, 8093 Zurich, Switzerland, \\
  Department of Mathematical Sciences, University of Copenhagen, Universitetsparken 5, 2100 Copenhagen, Denmark}
\email{christandl@math.ku.dk}

\author{M.~Burak \c{S}ahino\u{g}lu}
\address{Institute for Theoretical Physics, ETH Zurich, Wolfgang--Pauli--Strasse 27, 8093 Zurich, Switzerland, \\
  Vienna Center for Quantum Technology, University of Vienna, Boltzmanngasse 5, 1090 Vienna, Austria}
\email{burak.sahinoglu@univie.ac.at}

\author{Michael Walter}
\address{Institute for Theoretical Physics, ETH Zurich, Wolfgang--Pauli--Strasse 27, 8093 Zurich, Switzerland, \\
  Stanford Institute for Theoretical Physics, Stanford University, Stanford, CA 94305}
\email{michael.walter@stanford.edu}

\begin{abstract}
We prove that the asymptotic behavior of the recoupling coefficients of the symmetric group is characterized by a quantum marginal problem -- namely, by the existence of quantum states of three particles with given eigenvalues for their reduced density operators.
This generalizes Wigner's observation that the semiclassical behavior of the $6j$-symbols for $\SU(2)$ -- fundamental to the quantum theory of angular momentum -- is governed by the existence of Euclidean tetrahedra.
As a corollary, we deduce solely from symmetry considerations the strong subadditivity property of the von Neumann entropy. 
Lastly, we show that the problem of characterizing the eigenvalues of partial sums of Hermitian matrices arises as a special case of the quantum marginal problem.
We establish a corresponding relation between the recoupling coefficients of the unitary and symmetric groups, generalizing a classical result of Littlewood and Murnaghan.
\end{abstract}

\maketitle

\section{Introduction}

When does there exist a multi-particle quantum state compatible with a given set of reduced states?
Long recognized for its importance in many-body quantum physics and quantum chemistry \cite{coleman63,colemanyukalov00}, this \emph{quantum marginal problem} has seen significant progress in recent years in the context of quantum information theory through the discovery of an underlying group theoretic structure \cite{christandlmitchison06, klyachko04,daftuarhayden04}.
In particular, a complete solution of its most fundamental version, where the given reduced states are those of the individual particles, has been obtained \cite{klyachko04,daftuarhayden04}, and a firm connection to the classification of multiparticle entanglement has been established \cite{walter13, sawicki12, sawicki12pra}. Meanwhile, it has been understood that the general problem is computationally hard, even for a quantum computer \cite{liu06,liuchristandlverstraete07}, and only sufficient conditions are known. The strong subadditivity of the von Neumann entropy is perhaps the most important such condition, and an indispensible tool in quantum statistical physics and quantum information theory \cite{liebruskai73, ohyapetz04}.
The discovery of any further entropy inequality would be considered a major breakthrough \cite{pippenger03} (cf.~\cite{lindenwinter05,cadneylindenwinter12,LindenMatusRuskaiEtAl13,gross2013stabilizer,hayden2013holographic,bao2015holographic} for recent progress).

In this work we unveil a novel link between the existence of multi-particle quantum states with given marginal eigenvalues and the representation theory of the symmetric group $S_k$. For this, we consider the \emph{recoupling coefficients} of $S_k$, where now $k$ plays the role of the semiclassical parameter. We then find that the coefficients' norm decreases at most polynomially for a sequence of Young diagrams of $k \rightarrow \infty$ boxes converging to the eigenvalues of a given tripartite quantum state $\rho_{ABC}$ and its reduced states $\rho_A$, $\rho_B$, $\rho_C$, $\rho_{AB}$, $\rho_{BC}$; conversely, if there exists no such quantum state then the sequence decreases exponentially in norm.

Our result directly relates to the recent efforts on the quantum marginal problem and the understanding of quantum entropy. In particular, it extends the characterisation of the triples $\rho_A$, $\rho_C$, $\rho_{AC}$ by the Kronecker coefficient of the symmetric group \cite{christandlmitchison06, klyachko04, daftuarhayden04, christandlharrowmitchison07, christandldorankousidiswalter12}. The power of this extension is illustrated by the fact that the symmetries of the recoupling coefficients alone imply the strong subadditivity and weak monotonicity of von Neumann entropy (note that entropy is only a function of the eigenvalues). Our result generalizes directly to an arbitrary number of particles and linearly many reduced states, and may thus be regarded as a partial quantum-mechanical version of Chan and Yeung's description of the set of compatible Shannon entropies in terms of group theory \cite{chanyeung02}.
In fact, our work suggests a new route towards establishing further entropy inequalities by exploiting the symmetries of higher-order representation-theoretic objects.

Our work is inspired by Wigner's seminal work on the semiclassical behavior of quantum spins, which are described by the representation theory of the group $\SU(2)$ \cite{wigner59}.
The recoupling coefficients of $\SU(2)$, known as the \emph{Wigner $6j$-symbols} in their rescaled, more symmetric form \cite{wigner59}, describe the relation between individual spins $j_1$, $j_2$, $j_3$, their total spin $j_{123}$ and the intermediate spins $j_{12}$ and $j_{23}$ (the Racah W-coefficients \cite{racah42} are also closely related).
As first noted by Wigner, there is a dichotomy similar to our result in the semiclassical limit where all spins are simultaneously large: the $6j$-symbol decays polynomially if there exists a tetrahedron with side lengths $j_1$, $j_2$, $j_3$, $j_{12}$, $j_{23}$, $j_{123}$, and exponentially otherwise \cite{wigner59}.
That the asymptotics are in both cases guided by the existence of a geometric object---for Wigner, a tetrahedron with certain side lengths, for us, a quantum state with certain spectral properties---is not accidental. Via Schur-Weyl duality, our limit $k \rightarrow \infty$ can similarly be understood as a semiclassical limit. 
We furthermore describe how to construct for every tetrahedron a tripartite quantum state in a faithful, i.e.\ sidelength-encoding way, and for every Wigner $6j$-symbol we construct a corresponding recoupling coefficient of the symmetric group .
In this way, Wigner's problem and its generalization to $\SU(d)$ can be understood as special case of the quantum marginal problem.

We speculate that this surprising connection between tetrahedra and quantum states as well as between Wigner $6j$-symbols and symmetric group recoupling coefficients may help to understand and connect the study of spin foams and spin networks in the context of quantum gravity \cite{ooguri92, reisenbergrovelli97, freidellouapre03, barrettsteele03, gurau, Aquilanti} and condensed matter physics \cite{LevinWen05}
as well as topological quantum computing \cite{jordan08,jordan10,akimotohayashi11}.
Preliminary versions of this work have appeared in \cite{masterthesis12,Walter14}.

\section{Recoupling coefficients}
\label{sec:recoupling}
The finite-dimensional irreducible representations of the symmetric group $S_k$ are labelled by \emph{Young diagrams}, that is, ordered partitions $\lambda_1 \geq \ldots \geq \lambda_l > 0$ of $\sum_i \lambda_i = k$.
We may think of $k$ as the number of boxes of the Young diagram and of $l$ as the number of rows.
We write $\lambda \vdash k$ for such a partition, $[\lambda]$ for the associated irreducible unitary representation of $S_k$, and denote the dimension of the latter by $d_\lambda := \dim [\lambda]$.
Any finite-dimensional representation $V$ of $S_k$ can be decomposed into a direct sum of irreducible representations, and if $V$ is a unitary representation then this decompositoni can also be made unitary.
Concretely, consider the space of $S_k$-linear maps $\Ha^V_\lambda := \Hom_{S_k}([\lambda], V)$ and equip $\Ha^V_\lambda$ with the re-scaled Hilbert-Schmidt inner product $\braket{\psi, \phi}_\lambda := \tr \psi^\dagger \phi / d_\lambda$.
Then the canonical maps $\Phi^V_\lambda \colon [\lambda] \otimes \Ha^V_{\lambda} \rightarrow V, v \otimes \phi \mapsto \phi(v)$ are $S_k$-linear isometries that can be assembled to a unitary isomorphism $\bigoplus_{\lambda \vdash k} [\lambda] \otimes \Ha^V_{\lambda} \cong V$.

In particular, we may decompose a tensor product $[\alpha] \otimes [\beta]$ of two irreducible representations.
This results in the so-called \emph{Clebsch-Gordan isomorphism} of the symmetric group $S_k$,
\begin{equation}
\label{clebsch gordan iso}
	\bigoplus_{\lambda \vdash k} [\lambda] \otimes \Ha^{\alpha \beta}_{\lambda} \rightarrow [\alpha] \otimes [\beta].
\end{equation}
Its components are $S_k$-linear isometries and will be denoted by
\begin{equation}
  \label{clebsch gordan}
  \CG^{\alpha\beta}_\lambda \colon [\lambda] \otimes \Ha^{\alpha \beta}_{\lambda} \rightarrow [\alpha] \otimes [\beta].
\end{equation}
The dimension of $\Ha^{\alpha\beta}_\lambda$ is known as the \emph{Kronecker coefficient} $g_{\alpha\beta\gamma}$; it is fully symmetric in its three indices since the representations of $S_k$ are self-dual.

We now consider a triple tensor product. Since the tensor product is associative, we have
\[ \bigl( [\alpha] \otimes [\beta] \bigr) \otimes [\gamma] \cong [\alpha] \otimes \bigl( [\beta] \otimes [\gamma] \bigr). \]
Decomposing accordingly using \cref{clebsch gordan iso}, we obtain an isomorphism
\begin{equation}
  \label{6j decomposition}
  \bigoplus_{\lambda \vdash k} [\lambda] \otimes \bigoplus_{\mu \vdash k} \Ha_{\lambda}^{\mu \gamma} \otimes \Ha_{\mu}^{\alpha \beta}
  \cong
  \bigoplus_{\lambda \vdash k} [\lambda] \otimes \bigoplus_{\nu \vdash k} \Ha_{\lambda}^{\alpha \nu} \otimes \Ha_{\nu}^{\beta \gamma}.
\end{equation}
By Schur's lemma, this allows us to identify the multiplicity spaces for each fixed $\lambda$,
\begin{equation}
\label{6j decomposition for fixed lambda}
	\bigoplus_\mu \Ha_{\lambda}^{\mu \gamma} \otimes \Ha_{\mu}^{\alpha \beta} \cong \bigoplus_\nu \Ha_{\lambda}^{\alpha \nu} \otimes \Ha_{\nu}^{\beta \gamma}.
\end{equation}

\begin{dfn}[Recoupling coefficients]
  \label{def 6j definition}
	The \emph{recoupling coefficients} of the symmetric group are defined to be the components of the isomorphism~\eqref{6j decomposition for fixed lambda} for fixed $\mu$ and $\nu$, denoted by
\begin{equation}
\label{6j definition}
  \sixj\alpha\beta\gamma\mu\nu\lambda \colon
  \Ha_{\lambda}^{\mu \gamma} \otimes \Ha_{\mu}^ {\alpha \beta} \rightarrow \Ha_{\lambda}^{\alpha \nu} \otimes \Ha_{\nu}^{\beta \gamma}.
\end{equation}
\end{dfn}

In other words, they are defined by the relation
\begin{equation}
\label{6j via clebsch gordan}
  \id_{\lambda} \otimes \sixj \alpha \beta \gamma \mu \nu \lambda
  =
  \left(
  \left( \id_{\alpha} \otimes \CG^{\beta\gamma}_\nu \right)
  \CG^{\alpha\nu}_\lambda
  \right)^\dagger
  \left( \CG^{\alpha\beta}_\mu \otimes \id_{\gamma} \right)
  \CG^{\mu\gamma}_\lambda,
\end{equation}
in terms of the Clebsch-Gordan maps defined in \cref{clebsch gordan}; $\id_\lambda$ denotes the identity map on the representation $[\lambda]$.

In contrast to the case of $\SU(2)$, these recoupling coefficients are linear maps rather than scalars, since the ``Clebsch-Gordan series'' for $S_k$ is \emph{not} multiplicity-free in general.
Their size is thus best measured by an operator norm, and it will be convenient to employ the \emph{Hilbert-Schmidt norm}, $\normHS{X}^2 := \tr X^\dagger X$ as well as the \emph{operator norm} $\norm{X}_\infty := \sup_{\norm \psi = 1} \norm{X \psi}$.

Since the dimension of the multiplicity spaces $\Ha^{\alpha\beta}_\lambda$, i.e., the Kronecker coefficients $g_{\alpha\beta\gamma}$, are at most $\poly(k)$,%
\footnote{Here and throughout this text, $\poly(k)$ denotes polynomials in $k$ that depend only on the maximal number of rows of the Young diagrams. Equivalently, they depend only on $a$, $b$ and $c$---the dimensions of the local Hilbert spaces $\CC^a$, $\CC^b$ and $\CC^c$ in \cref{tripartite schur weyl}.} both norms are in fact polynomially equivalent:
\begin{equation}
  \label{6j lower upper}
  \norm{\begin{bmatrix} \alpha & \beta & \mu \\ \gamma & \lambda & \nu \end{bmatrix}}_\infty \\
  \leq \normHS{\begin{bmatrix} \alpha & \beta & \mu \\ \gamma & \lambda & \nu \end{bmatrix}}\\
  \leq \poly(k) \norm{\begin{bmatrix} \alpha & \beta & \mu \\ \gamma & \lambda & \nu \end{bmatrix}}_\infty.
\end{equation}

\subsection*{Schur-Weyl duality}

As a first step towards connecting to the properties of quantum states, we now consider
the vector space $V = \left( \CC^d \right)^{\otimes k}$.
The symmetric group $S_k$ acts on $V$ by permuting the tensor factors and the special unitary group $\SU(d)$ acts diagonally; both actions commute.
The corresponding decomposition into irreducible $S_k \times \SU(d)$-representations is known as \emph{Schur-Weyl duality} (e.g.~\cite{goodmanwallach09}),
\begin{equation}
\label{Schur-Weyl duality}
  \left( \CC^d \right)^{\otimes k} \cong
  \; \smashoperator{\bigoplus_{\lambda \vdash_d k}} \; [\lambda] \otimes V^d_\lambda
\end{equation}
Here, $V^d_\lambda$ denotes the irreducible $\SU(d)$-representation with highest weight $\lambda$, and the direct sum runs over all Young diagrams $\lambda$ with $k$ boxes and no more than $d$ rows.
In the following we denote by $P^d_\lambda$ the orthogonal projector onto a direct summand $[\lambda] \otimes V^d_\lambda$ in~\eqref{Schur-Weyl duality}.

We now consider $\CC^{abc} \cong \CC^a \otimes \CC^b \otimes \CC^c$. Applying Schur-Weyl duality separately to the $k$-th tensor powers of $\CC^a$, $\CC^b$ and $\CC^c$, respectively, and applying \cref{6j decomposition}, we find that
\begin{equation}
  \label{tripartite schur weyl}
  \begin{aligned}
  \left( \CC^{abc} \right)^{\otimes k} &\cong
  \; \smashoperator{\bigoplus_{\alpha \vdash_a k, \beta \vdash_b k, \gamma \vdash_c k}} \; [\alpha] \otimes [\beta] \otimes [\gamma] \otimes V^a_\alpha \otimes V^b_\beta \otimes V^c_\gamma \\
  &\cong
  \;\smashoperator{\bigoplus_{\alpha \vdash_a k, \dots, \lambda \vdash_{abc} k}}\;
  [\lambda] \otimes \Ha^{\mu \gamma}_{\lambda} \otimes \Ha^{\alpha \beta}_{\mu} \otimes V^a_\alpha \otimes V^b_\beta \otimes V^c_\gamma\\
  &\cong
  \;\smashoperator{\bigoplus_{\alpha \vdash_a k, \dots, \lambda \vdash_{abc} k}}\;
  [\lambda] \otimes \Ha_{\lambda}^{\alpha \nu} \otimes \Ha_{\nu}^{\beta \gamma} \otimes V^a_\alpha \otimes V^b_\beta \otimes V^c_\gamma.
  \end{aligned}
\end{equation}
We denote by
\begin{equation}
\label{eq:PQ definition}
\begin{aligned}
	P &:= P^{\alpha\beta\gamma}_{\lambda,\mu} := (P_\alpha^a \otimes P_\beta^b \otimes P_\gamma^c) (P_\mu^{ab} \otimes P_\gamma^c) P_\lambda^{abc} \\
	Q &:= Q^{\alpha\beta\gamma}_{\lambda,\nu} := (P_\alpha^a \otimes P_\beta^b \otimes P_\gamma^c) (P_\alpha^a \otimes P_\nu^{bc}) P_\lambda^{abc}
\end{aligned}
\end{equation}
the orthogonal projectors onto the corresponding direct summands in the second and third line of \cref{tripartite schur weyl}, respectively.
Observe that each is defined as a product of commuting projectors.
The following lemma connects the operator norm of their product to the operator norm of the corresponding recoupling coefficient.

\begin{lem}
  \label{lem PQ vs 6j}
 Then
\begin{equation}
  \label{PQ vs 6j}
  \norm{P Q}_\infty =
  \norm{Q P}_\infty =
  \norm{\sixj \alpha \beta \gamma \mu \nu \lambda}_\infty .
\end{equation}
\end{lem}

\begin{proof}
By \cref{6j via clebsch gordan}, the recoupling coefficients satisfy the identity
\begin{equation*}
  \id_{\lambda} \otimes \sixj \alpha \beta \gamma \mu \nu \lambda = F^\dagger E,
\end{equation*}
where $E := \left( \CG^{\alpha\beta}_\mu \otimes \id_{\gamma} \right) \CG^{\mu\gamma}_\lambda$ and $F := (\id_\alpha \otimes \CG^{\beta\gamma}_\nu) \CG^{\alpha\nu}_\lambda$ are compositions of Clebsch-Gordan isometries as defined in \cref{clebsch gordan}.
In particular,
\begin{equation*}
  \norm{\sixj \alpha \beta \gamma \mu \nu \lambda}_\infty =
  \norm{F^\dagger E}_\infty =
  \norm{F F^\dagger E E^\dagger}_\infty.
\end{equation*}
The last equality holds because both $E$ and $F$ are isometries.

Now note that $E E^\dagger$ is precisely equal to the orthogonal projector onto
\begin{equation*}
  [\lambda] \otimes \Ha_\lambda^{\mu\gamma} \otimes \Ha_\mu^{\alpha\beta} \subseteq [\alpha] \otimes [\beta] \otimes [\gamma].
\end{equation*}
On the other hand, $P$ as defined in \cref{eq:PQ definition} is the projector onto
\begin{align*}
  &[\lambda] \otimes \Ha_\lambda^{\mu\gamma} \otimes \Ha_\mu^{\alpha\beta} \otimes V^a_\alpha \otimes V^b_\beta \otimes V^c_\gamma \\
  \subseteq
  &[\alpha] \otimes [\beta] \otimes [\gamma] \otimes V^a_\alpha \otimes V^b_\beta \otimes V^c_\gamma
  \subseteq
  \left( \mathbb C^{abc} \right)^{\otimes k}.
\end{align*}
It follows that $P = P = E E^\dagger \otimes \id_{V^a_\alpha \otimes V^b_\beta \otimes V^c_\gamma}$, where we slightly abuse notation by considering the right-hand side as an opreator on $(\CC^{abc})^{\otimes k}$.
Likewise, we find that $Q = F F^\dagger \otimes \id_{V^a_\alpha \otimes V^b_\beta \otimes V^c_\gamma}$. Together we conclude that
\begin{equation*}
  \norm{F F^\dagger E E^\dagger}_\infty = \norm{Q P}_\infty. \qedhere
\end{equation*}
\end{proof}




Schur-Weyl duality also leads to an alternative definition of the recoupling coefficients in terms of unitary groups (see discussion at the end of \cref{sec:tripartite}).

\section{Recoupling coefficients and tripartite quantum marginals}
\label{sec:tripartite}

Let $\rho_{ABC}$ be a \emph{density operator}, i.e., a positive semidefinite operator of trace one, on $\CC^a \otimes \CC^b \otimes \CC^c$.
We consider the reduced density operators $\rho_{AB}=\tr_{C} (\rho_{ABC})$, $\rho_A= \tr_{BC} (\rho_{ABC})$ etc.,%
\footnote{Given a density operator $\rho_{12}$ on a tensor product Hilbert space $\mathcal H_1 \otimes \mathcal H_2$, the reduced density operator $\rho_1 = \tr_2 (\rho_{12})$ is uniquely defined by the property that $\tr \rho_1 X_1 = \tr \rho_{12} (X_1 \otimes \id_{\mathcal H_2})$ for all operators $X_1$ on $\mathcal H_1$.}
and denote by $r_{ABC}$, $r_{AB}$, $r_A$, etc., the corresponding vectors of eigenvalues (each ordered non-increasingly, e.g.~$r_{ABC,1} \geq r_{ABC,2} \geq \ldots$).
Let us also define the normalisation of a Young diagram $\lambda \vdash k$ by $\bar\lambda := \lambda / k$.
Then we have the following theorem:

\begin{thm}
  \label{main result}
  If there exists a quantum state $\rho_{ABC}$ with eigenvalues $r_A$, $r_B$, $r_C$, $r_{AB}$, $r_{BC}$, $r_{ABC}$ then there exist Young diagrams $\alpha, \beta, \gamma, \mu, \nu, \lambda \vdash k$ with $k \rightarrow \infty$ boxes and at most $a$, $b$, etc.\ rows such that
    \begin{align}\label{assumption-Conv}
      \lim_{k \rightarrow \infty}(\bar\alpha, \bar\beta, \bar\gamma, \bar\mu, \bar\nu, \bar\lambda)= (r_A, r_B, r_C, r_{AB}, r_{BC}, r_{ABC})
    \end{align}
    and
    \begin{equation} \label{assumptionHS}
      \normHS{\begin{bmatrix} \alpha & \beta & \mu \\ \gamma & \lambda & \nu \end{bmatrix}}
      \geq \frac 1 {\mathrm{poly}(k)}.
    \end{equation}
  Conversely, if $(r_A, r_B, r_C, r_{AB}, r_{BC}, r_{ABC})$ is not associated to any tripartite density operator then for every sequence of Young diagrams satisfying \cref{assumption-Conv} we have
  \begin{equation}
    \label{converse}
    \normHS{\begin{bmatrix} \alpha & \beta & \mu \\ \gamma & \lambda & \nu \end{bmatrix}} \leq \exp\bigl(-\Omega(k)\bigr).
  \end{equation}
\end{thm}

\begin{proof}
For both directions of the proof we use the spectrum estimation theorem \cite{keylwerner01} (cf.~\cite{AlRuSa87, HayMat02, christandlmitchison06}), which says that $k$ copies of a density operator $\rho$ on some $\CC^d$ are mostly supported on the subspaces $[\lambda] \otimes V^d_\lambda$ satisfying $\bar\lambda = \lambda/k \approx \spec \rho = r$. More precisely,
\begin{equation}
\label{keyl werner}
\tr(P_\lambda^d \rho^{\otimes k}) \leq \poly(k) \exp\bigl(-k \norm{\bar{\lambda} - r}_1^2 / 2\bigr)
\end{equation}
where $\norm{x}_1= \sum_i \abs{x_i}$ is the $\ell_1$-norm.

We start with the proof of the ``if'' statement.
Define $\tilde P$ as the sum of the projectors $P_{\alpha\beta\gamma}^{\lambda,\mu}$ -- definend as in~\eqref{eq:PQ definition} -- for which $\norm{\bar{\alpha}- r_A}_1 \leq \delta$, $\norm{\bar{\beta} - r_B}_1 \leq \delta$, etc.; $\tilde Q$ is defined accordingly.
By~\eqref{keyl werner} and the fact that there are only $\poly(k)$ many Young diagrams 
\begin{equation*}
\tr(\tilde P \rho^{\otimes k}_{ABC}) \geq 1 - \varepsilon, \quad \tr(\tilde Q \rho^{\otimes k}_{ABC}) \geq 1 - \varepsilon,
\end{equation*}
where $\varepsilon = \poly(k) \exp(-k \delta^2 / 2)$.
Now we use
\begin{equation*}
  \abs{\tr(\tilde P \tilde Q \sigma)} \geq \tr(\tilde P \sigma) - \sqrt{1 - \tr(\tilde Q \sigma)},
\end{equation*}
which holds for arbitrary projectors $\tilde P$, $\tilde Q$ and quantum states $\sigma$\footnote{For pure states $\sigma = \ket\phi\!\bra\phi$, $\tr(PQ \sigma)= \braket{\phi | PQ | \phi} \geq \braket{\phi | P |\phi} - \abs{\braket{\phi | P (1 - Q) | \phi}} \geq \braket{\phi | P | \phi} - \norm{(1 - Q) \ket\phi} = \tr(P\sigma) - \sqrt{1 - \tr(Q\sigma)}$ by the Cauchy-Schwarz inequality; the general statement follows by considering a purification of $\sigma$.} and obtain
\begin{equation*}
  \norm{\tilde P \tilde Q}_\infty \geq
  \abs{\tr(\tilde P \tilde Q \rho_{ABC}^{\otimes k})} \geq 1 - 2 \sqrt{\varepsilon}.
\end{equation*}
Using \cref{lem PQ vs 6j}, \cref{6j lower upper} and the triangle inequality, we find that
\begin{equation*}
\displaystyle\sum_{\alpha, \beta, \gamma, \mu, \nu, \lambda}\normHS{\begin{bmatrix} \alpha & \beta & \mu \\ \gamma & \lambda & \nu \end{bmatrix}} \geq 1- 2 \sqrt\varepsilon,
\end{equation*}
where the sum extends over Young diagrams with $k$ boxes and at most $a$, $b$, etc.\ rows, whose normalisation is close to the eigenvalues associated to $\rho_{ABC}$ as specified above. Since the number of terms in the sum is again upper-bounded by $\poly(k)$, we can find sequences of Young diagrams satisfying \cref{assumption-Conv} and \cref{assumptionHS}.

We now prove the converse statement.
For this, we consider a sequence of Young diagrams \cref{assumption-Conv} whose limit $(r_A, r_B, \ldots)$ is \emph{not} associated to any tripartite density operator on $\CC^{abc}$.
For each $k$, we choose quantum states $\sigma_{(ABC)^k}$ such that
\[
  \norm{P Q}_\infty^2
  =
  \norm{(QP)(PQ)}_\infty
  =
  \tr(P Q \sigma_{(ABC)^k} Q P ).
\]
Both projectors $P$ and $Q$ commute with the action of $S_k$ on $(\CC^{abc})^{\otimes k}$, so we may assume that the same is true for $\sigma_{(ABC)^k}$.
Then we can use the bound
\begin{equation}
\label{eq:postsel}
  \sigma_{(ABC)^k}\leq \mathrm{poly}(k) \int d\rho_{ABC} \, \rho_{ABC}^{\otimes k},
\end{equation}
where $d\rho_{ABC}$ is the Hilbert-Schmidt probability measure on the set of density operators on $\CC^{abc}$
(see, e.g., \cite[Lemma 1]{christandl_renner_2012_reliable}). 
The right-hand side of~\eqref{eq:postsel} commutes with the action of $S_k$ as well as with unitaries of the form $U^{\otimes k}$, $U \in \SU(abc)$.
In view of Schur-Weyl duality~\eqref{Schur-Weyl duality}, Schur's lemma implies that it is a lienar combination of the isotypical projectors $P^{abc}_\lambda$ and therefore commutes with both $P$ and $Q$.
It follows that 
\begin{align*}
	\norm{PQ}_\infty^2 &\leq \mathrm{poly}(k) \tr(P Q \int d\rho_{ABC} \, \rho_{ABC}^{\otimes k} Q P) \\
	&= \mathrm{poly}(k) \tr(P \int d\rho_{ABC} \, \rho_{ABC}^{\otimes k} Q) \\
	&= \mathrm{poly}(k) \int d\rho_{ABC} \, \tr(P \rho_{ABC}^{\otimes k} Q)
\end{align*}
by linearity and cyclicity of the trace.
Since $d\rho_{ABC}$ is a probability measure, it follows, together with \cref{6j lower upper,PQ vs 6j}, that for each $k$ there exists at least one quantum state $\rho_{ABC, k}$ on $\CC^{abc}$ such that
\begin{align*}
  \normHS{\begin{bmatrix} \alpha & \beta & \mu \\ \gamma & \lambda & \nu \end{bmatrix}}^2
  \leq
  \mathrm{poly}(k) \tr(P \rho_{ABC, k}^{\otimes k} Q).
\end{align*}
Using the H\"older inequality, the right-hand side can be upper-bounded by the square roots of each of the six traces $\tr(P_\alpha^a \rho_A^{\otimes k})$, $\tr(P_\beta^b \rho_B^{\otimes k})$, etc., which in turn can be upper-bounded via \cref{keyl werner}. Thus we find
\begin{equation*}
  \normHS{\begin{bmatrix} \alpha & \beta & \mu \\ \gamma & \lambda & \nu \end{bmatrix}}^2 \leq \poly(k) \exp(-k D^2 /4),
\end{equation*}
where $D := \min_{(s_A, s_B, \dots)} \max \{ \norm{\bar\alpha - s_A}_1, \norm{\bar\beta - s_B}_1, \dots \}$ quantifies the distance of $(r_A, r_B, \ldots)$ to the closed set of spectra $(s_A, s_B, \ldots)$ that correspond to tripartite quantum states on $\CC^{abc}$.
By assumption, $D > 0$, so we get the exponential decay asserted in the theorem.
\end{proof}

\Cref{main result} can be generalized to more than three parties by considering the following quantity: as in \cref{6j decomposition}, successively decompose a tensor product of irreducible representations in two inequivalent ways; the corresponding ``generalized recoupling coefficients'' then are the components of the resulting isomorphism for fixed intermediate labels $\mu_j$ and $\nu_k$, and an analogous result can be established for these coefficients, which are symmetric group counterparts of Wigner's $3nj$-symbols for $\SU(2)$.
Just as \cref{main result} does not cover the eigenvalues of $\rho_{AC}$, in general only of a linear number of the exponentially many reduced density operators can be controlled in this fashion (e.g., the nearest-neighbor reduced states in a linear chain of particles).
What the representation-theoretic quantities involved in controlling all marginal spectra should be is an intriguing question, with possible ramifications to the search for new entropy inequalities of the von Neumann entropy, as we detail in the following sections.

\medskip

For pure quantum states $\rho_{ABC}$, the Schmidt decomposition implies that necessarily $r_{AB} = r_C$ and $r_A = r_{BC}$.
Therefore, we can discard of the two-body spectra, and the problem reduces to a one-body quantum marginal problem.
On the level of representation theory, it suffices to consider single-row Young diagrams $\lambda = (k)$, corresponding to the trivial representation of $S_k$; hence, $\mu = \gamma$ and $\alpha = \nu$, and it can be shown easily that
\[ \normHS{\sixj \alpha \beta \gamma \mu \nu \lambda}^2 = \dim \Ha^{\alpha\beta}_\gamma = g_{\alpha\beta\gamma} \]
likewise reduces to a Kronecker coefficient of the symmetric group.

In this way, \cref{main result} specializes to the well-known relationship between the pure-state one-body quantum marginal problem and the asymptotics of the decomposition of tensor products of irreducible representations of the symmetric group~\cite{christandlmitchison06, klyachko04, christandlharrowmitchison07} (it also shows that recoupling coefficients can grow with $k$).
\Cref{main result} generalizes this relationship:
It shows that the overlap between two such decompositions -- as captured by the recoupling coefficients -- similarly characterizes the quantum marginal problem with two overlapping marginals.
It would be of great interest to find a geometric explanation of this result in the framework of geometric invariant theory, which might also lead to a more refined understanding of the asymptotics along the lines of~\cite{roberts1999} for Wigner's $6j$-symbols.
Mathematically, this is related to the ``intersection'' of moment maps, or to simultaneous Hamiltonian reduction for non-commuting group actions.

\medskip

We conclude with some remarks on the interpretation of \cref{main result} as a semiclassical limit.
In \cite{wigner59}, Wigner studied the asymptotics of the recoupling coefficients of $\SU(2)$ which can be defined in complete analogy to \cref{def 6j definition}. Given three particles of spin $j_A$, $j_B$, $j_C$ such that the total spin of the first two particles is $j_{AB}$ and of all three particles $j_{ABC}$, the absolute value squared of the $\SU(2)$ recoupling coefficient can be interpreted as the probability of observing that particles two and three have total spin $j_{BC}$. In the semiclassical limit of simultaneously large spins, Wigner showed that this probability oscillates around the inverse volume of the tetrahedron whose edges have length equal to the six spins---if such a tetrahedron exists.
In particular, it then decays polynomially with $j$.
If no such tetrahedron exists then the recoupling coefficient decays exponentially.
This result is understood to mean that ``classical'' configurations are exponentially more likely than all others in the limit of large quantum numbers.
A more precise formula has been given by Ponzano and Regge \cite{PonzanoRegge} and only fully proved in \cite{roberts1999}.

The recoupling coefficients for the symmetric groups $S_k$ that we consider in this chapter can also be defined in terms of unitary groups. This follows from Schur--Weyl duality \eqref{Schur-Weyl duality}, which implies that the projectors $P^a_\alpha$, $P^b_\beta$, etc.\ in \eqref{eq:PQ definition} can be equivalently defined as the isotypical projectors for the unitary groups $\SU(a)$, $\SU(b)$, etc.
From this perspective, the decompositions in \eqref{tripartite schur weyl} arise by restricting the $\SU(abc)$-representation $(\CC^{abc})^{\otimes k}$ to $\SU(a) \times \SU(b) \times \SU(c)$ via either of the ``intermediate subgroups'' $\SU(ab) \times \SU(c)$ or $\SU(a) \times \SU(bc)$.
Thus, it is suggestive to consider the number of boxes $k$ in the Young diagrams as the semiclassical parameter in our setup. In this sense, tripartite quantum states $\rho_{ABC}$ are the formal analogues of Wigner's tetrahedra---they are the geometric objects that describe the ``classical'' configurations, corresponding to polynomial decay in the limit $k \rightarrow \infty$.


\section{Symmetries of the recoupling coefficients}

In this section we rewrite the Hilbert-Schmidt norm of the recoupling coefficients in a way that makes manifest its symmetries.
Both the strong subadditivity property of the von Neumann entropy as well as its weak monotonicity property can then be understood in terms of these symmetries and \cref{main result} (see \cref{sec:ssa}).
We first give a diagrammatic argument using the graphical calculus for symmetric monoidal categories~\cite{Coecke10,Selinger11,Turaev10}.
An alternative, purely algebraic proof, is postponed to the end of this section.

Recall from \cref{sec:recoupling} that the multiplicity spaces $\Ha^{\alpha\beta}_\lambda$ are given by the space of $S_k$-linear maps from $[\lambda]$ to $[\alpha] \otimes [\beta]$.
In each multiplicity space, let us choose maps $\Phi^{\alpha\beta}_{\lambda,i}$ that form an orthonormal basis with respect to the inner product $\braket{\psi, \phi}_\lambda = \tr \psi^\dagger \phi / d_\lambda$ introduced in \cref{sec:recoupling}.
We will represent these maps graphically by
\begin{equation*}
  \Phi^{\alpha\beta}_{\lambda,i} = \strongsixjCG{$\lambda$}{$\alpha$}{$\beta$}{$i$}
  \quad \text{and} \quad
  (\Phi^{\alpha\beta}_{\lambda,i})^{\dagger}= \strongsixjCGadjoint{$\lambda$}{$\alpha$}{$\beta$}{$i$}
\end{equation*}
in the graphical calculus.
The maps $\Phi^{\alpha\beta}_{\lambda,i}$ are nothing but components of the Clebsch--Gordan isometries $\Phi_\lambda^{\alpha\beta}$ defined in \cref{clebsch gordan}.
We thus obtain the following graphical expression for the matrix elements of the recoupling coefficients with respect to the bases fixed above from \eqref{6j via clebsch gordan} by taking a trace over $[\lambda]$:

\begin{equation}
\label{6j graphical expression}
  \begin{bmatrix} \alpha & \beta & \mu \\ \gamma & \lambda & \nu \end{bmatrix}^{kl}_{ij}
  =
  \dfrac{1}{d_\lambda} \strongsixjAsymmetric
\end{equation}

Our goal is to transform the right-hand side expression in \eqref{6j graphical expression} into a form that renders its symmetries apparent.
For this, we recall that the irreducible representations of the symmetric group are self-dual, i.e., $[\lambda] \cong [\lambda]^*$, because they can be defined over the reals.
It follows that there exists a single copy of the trivial representation $\mathbf{1}$ in each tensor product $[\lambda] \otimes [\lambda]$, i.e., $\Ha^{\lambda=\lambda}_{\mathbf 1}$ is one-dimensional.
We shall denote the corresponding basis vector by
\begin{equation}
\label{Special CG}
  \begin{tikzpicture}[baseline=-0.1cm]
    \draw[->-=.6] (-0.1,0) -- (-0.8, 0);
    \draw[->-=.6] (0.1,0) -- (0.8,0);
    \draw[epr] (0,0) circle[radius=0.1];
    \draw (-0.9,0) node[left] {$\lambda$};
    \draw (0.9,0) node[right] {$\lambda$};
  \end{tikzpicture}
  := \strongsixjCG{$\mathbf{1}$}{$\lambda$}{$\lambda$}{},
\end{equation}
omitting the leg corresponding to the identity object $\mathbf 1$ as is usual in the graphical calculus.
It can be concretely written as a maximally entangled state $\sum_i \ket{\lambda,i} \otimes \ket{\lambda,i} / {\sqrt {d_\lambda}}$ in any real orthonormal basis $\ket{\lambda,i}$ of $[\lambda]$ (i.e., in a basis such that $S_k$ acts by real orthogonal matrices).
We denote the adjoint of~\eqref{Special CG} by reversing arrows.
It is then easy to see that we have the ``teleportation identity''
\begin{equation}
\label{Two bubbles}
  \begin{tikzpicture}[baseline=-0.1cm]
    \draw[->-=.6] (-0.35,0) -- (0.35, 0);
    \draw[-<-=.35] (0.55,0) -- (1.25,0);
    \draw[->-=.6] (-0.55,0) -- (-1.25,0);
    \draw[adjoint epr] (0.45,0) circle[radius=0.1];
    \draw[epr] (-0.45,0) circle[radius=0.1];
    \draw (-1.35,0) node[left] {$\lambda$};
    \draw (1.35,0) node[right] {$\lambda$};
  \end{tikzpicture}
  =
  \frac 1 {d_\lambda}
  \begin{tikzpicture}[baseline=-0.1cm]
    \draw[-<-=.4] (0,0) -- (0.8,0);
    \draw (0,0) node[left] {$\lambda$};
    \draw (0.8,0) node[right] {$\lambda$};
  \end{tikzpicture}
  .
\end{equation}

We can use \eqref{Special CG} and its adjoint to raise and lower indices, i.e., to reverse the direction of arrows. We thus obtain the following important property of the Clebsch--Gordan isometries (cf.\ \cite[(7-205a)]{Hamermesh89}):

\begin{lem}
  \label{lem:invariant bases}
  Both sets
  \begin{equation*}
    \left\{ \strongsixjCGinvariant{$i$} \right\}
    \quad\text{and}\quad
    \left\{
      \sqrt{ \frac {d_\alpha d_\beta} {d_\lambda} }
      \strongsixjCGadjointinvariant{$i$}
    \right\}
  \end{equation*}
  form orthonormal bases of the space $([\alpha] \otimes [\beta] \otimes [\lambda])^{S_k}$ of $S_k$-invariant vectors in the triple tensor product.
\end{lem}
\begin{proof}
  Since the dimensions of $([\alpha] \otimes [\beta] \otimes [\lambda])^{S_k}$ and of $\Ha^{\alpha\beta}_\lambda$ agree by self-duality of $[\lambda]$, it suffices to show that both sets of vectors are orthonormal.
  For the first set, observe that it follows from the teleportation identity \eqref{Two bubbles} that
  \begin{equation*}
    \left\langle \strongsixjCGinvariant{$i$} | \strongsixjCGinvariant{$i'$} \right\rangle
    =
    \begin{tikzpicture}[baseline=0cm]
      \draw[->-=.56] (-0.2,-0.6) -- (-0.2,0.6);
      \draw[->-=.56] (0.2,-0.6) -- (0.2,0.6);

      \draw[->-=.6] (0.8,-0.3) -- (0.8,0.3);
      \draw[->-=0.4] (0.8,0.8) -- (0.8,0.5);
      \draw[-<-=0] (0.8,-0.8) -- (0.8,-0.5);
      \draw[adjoint epr] (0.8,0.4) circle[radius=0.1];
      \draw[epr] (0.8,-0.4) circle[radius=0.1];

      \draw[adjoint] (0, 0.6) circle[radius= 0.2];
      \draw (0, -0.6) circle[radius= 0.2];
      \draw (0, 0.8) arc[radius= 0.4, start angle= 180, end angle= 0];
      \draw (0, -0.8) arc[radius= 0.4, start angle= -180, end angle= 0];
      \draw (0,0.6) node {$i$};
      \draw (0,-0.6) node {$i'$};
      \draw (-0.2, 0) node[left] {$\alpha$};
      \draw (0.2, 0) node[right] {$\beta$};
      \draw (0.2, 0.9) node[right] {$\lambda$};
    \end{tikzpicture}
    =
    \frac 1 {d_\lambda}
    \begin{tikzpicture}[baseline=0cm]
      \draw[->-=.56] (-0.2,-0.6) -- (-0.2,0.6);
      \draw[->-=.56] (0.2,-0.6) -- (0.2,0.6);
      \draw[->-=.56] (0.8,0.8) -- (0.8,-0.8);
      \draw[adjoint] (0, 0.6) circle[radius= 0.2];
      \draw (0, -0.6) circle[radius= 0.2];
      \draw (0, 0.8) arc[radius= 0.4, start angle= 180, end angle= 0];
      \draw (0, -0.8) arc[radius= 0.4, start angle= -180, end angle= 0];
      \draw (0,0.6) node {$i$};
      \draw (0,-0.6) node {$i'$};
      \draw (-0.2, 0) node[left] {$\alpha$};
      \draw (0.2, 0) node[right] {$\beta$};
      \draw (0.2, 0.9) node[right] {$\lambda$};
    \end{tikzpicture}
    \;.
  \end{equation*}
  This is in turn equal to
  \begin{equation*}
    \frac {\tr (\Phi^{\alpha\beta}_{\lambda,i})^\dagger \Phi^{\alpha\beta}_{\lambda,i'}} {d_\lambda}
    = \braket{\Phi^{\alpha\beta}_{\lambda,i}, \Phi^{\alpha\beta}_{\lambda,i'}}_\lambda
    = \delta_{i,i'},
  \end{equation*}
  since the $\Phi^{\alpha\beta}_{\lambda,i}$ form an orthonormal basis.

  For the second set, we find similarly that
  \begin{equation*}
    \left\langle \strongsixjCGadjointinvariant{$i$} | \strongsixjCGadjointinvariant{$i'$} \right\rangle
    =
    \begin{tikzpicture}[baseline=0cm]
      \draw[->-=.56] (0.8,-0.8) -- (0.8,0.8);

      \draw[->-=.85] (0.2, -0.1) -- (0.2,0.1);
      \draw[->-=.9] (0.2,0.6) -- (0.2,0.3);
      \draw[->-=.56] (0.2,-0.3) -- (0.2,-0.6);
      \draw[adjoint epr] (0.2,0.2) circle[radius=0.1];
      \draw[epr] (0.2,-0.2) circle[radius=0.1];

      \draw[->-=.85] (-0.2, -0.1) -- (-0.2,0.1);
      \draw[->-=.9] (-0.2,0.6) -- (-0.2,0.3);
      \draw[->-=.56] (-0.2,-0.3) -- (-0.2,-0.6);
      \draw[adjoint epr] (-0.2,0.2) circle[radius=0.1];
      \draw[epr] (-0.2,-0.2) circle[radius=0.1];

      \draw (0, 0.6) circle[radius= 0.2];
      \draw[adjoint] (0, -0.6) circle[radius= 0.2];
      \draw (0, 0.8) arc[radius= 0.4, start angle= 180, end angle= 0];
      \draw (0, -0.8) arc[radius= 0.4, start angle= -180, end angle= 0];
      \draw (0,0.6) node {$i$};
      \draw (0,-0.6) node {$i'$};
      \draw (-0.2, 0) node[left] {$\alpha$};
      \draw (0.2, 0) node[right] {$\beta$};
      \draw (0.2, 0.9) node[right] {$\lambda$};
    \end{tikzpicture}
    = \frac 1 {d_\alpha d_\beta}
    \begin{tikzpicture}[baseline=0cm]
      \draw[->-=.56] (0.8,-0.8) -- (0.8,0.8);
      \draw[->-=.56] (0.2,0.5) -- (0.2,-0.6);
      \draw[->-=.56] (-0.2,0.5) -- (-0.2,-0.6);
      \draw (0, 0.6) circle[radius= 0.2];
      \draw[adjoint] (0, -0.6) circle[radius= 0.2];
      \draw (0, 0.8) arc[radius= 0.4, start angle= 180, end angle= 0];
      \draw (0, -0.8) arc[radius= 0.4, start angle= -180, end angle= 0];
      \draw (0,0.6) node {$i$};
      \draw (0,-0.6) node {$i'$};
      \draw (-0.2, 0) node[left] {$\alpha$};
      \draw (0.2, 0) node[right] {$\beta$};
      \draw (0.2, 0.9) node[right] {$\lambda$};
    \end{tikzpicture}
    = \frac {d_\lambda} {d_\alpha d_\beta} \delta_{i,i'}.
  \end{equation*}
\end{proof}

We finally introduce the symmetric notation:
\begin{equation}
\label{invariant notation}
  \begin{tikzpicture}[baseline=0cm]
    \draw[->-=.56] (0,-0.2) -- (0,-0.7);
    \draw[->-=.56] (0.177,0.1) -- (0.766, 0.5);
    \draw[->-=.56] (-0.177,0.1) -- (-0.766,0.5);
    \draw (0,0) circle[radius=0.2];
    \draw (0,0) node {$i$};
    \draw (-0.766,0.5) node[above] {$\alpha$};
    \draw (0.766, 0.5) node[above] {$\beta$};
    \draw (0,-0.7) node[below] {$\lambda$};
  \end{tikzpicture}
  :=
  \strongsixjCGinvariant{$i$}
\end{equation}
We note that the vectors~\eqref{invariant notation} depend on the choice of arrow that was reversed.
However, by \cref{lem:invariant bases} any such choice gives rise to unitarily equivalent bases of the space of $S_k$-invariants!
We thus obtain the following result:

\begin{prp}
  \label{prp:invariant object}
  We have
  \begin{equation}
  \label{eq:invariant norm}
    \frac 1 {d_\mu d_\nu} \norm{\sixj \alpha \beta \gamma \mu \nu \lambda}_{\text{HS}}^2
  = d_\alpha d_\beta d_\gamma d_\lambda d_\mu d_\nu \, \norm{\strongsixjSymmetric}^2_2,
  \end{equation}
  where $\norm{(x_{i,j,k,l})}_2 := \sqrt{\sum_{i,j,k,l} \abs{x_{i,j,k,l}}^2}$ denotes the $\ell_2$-norm of a tensor with indices $i,j,k,l$.
  The right-hand side does not depend on the choice of arrow that was reversed in the definition~\eqref{invariant notation}.
\end{prp}
\begin{proof}
  Recall from \eqref{6j graphical expression} that
  \begin{equation*}
      \normHS{\sixj \alpha \beta \gamma \mu \nu \lambda}^2
    = \frac 1 {d_\lambda^2} \norm{\strongsixjAsymmetric}^2_2.
  \end{equation*}
  By inserting the teleportation identity \eqref{Two bubbles} once for each of the six arrows, we obtain
  \begin{equation*}
    \frac 1 {d_\lambda^2} d_\alpha^2 d_\beta^2 d_\gamma^2 d_\mu^2 d_\nu^2 d_\lambda^2 \norm{\strongsixjAsymmetricTeleported}^2_2.
  \end{equation*}
  By first applying the unitary transformation that relates the second orthonormal basis in \cref{lem:invariant bases} to the first (at the vertices \circled{$k$} and \circled{$l$}) and then using definition \eqref{invariant notation} (at all four vertices), this is in turn equal to
  \begin{align*}
    \frac 1 {d_\lambda^2} d_\alpha^2 d_\beta^2 d_\gamma^2 d_\mu^2 d_\nu^2 d_\lambda^2 \frac {d_\lambda} {d_\alpha d_\nu} \frac {d_\nu} {d_\beta d_\gamma} \norm{\strongsixjSymmetric}^2_2,
  \end{align*}
  as the $\ell_2$-norm is unitarily invariant.
  From this we obtain the desired expression by deforming the diagram and simplifying the prefactor.
  The unitary invariance of the $\ell_2$-norm also shows that the expression is insensitive to the choice of which arrow is reversed in \eqref{invariant notation}.
\end{proof}

The right-hand side of \eqref{eq:invariant norm} is the symmetric group analogue of a \emph{Wigner $6j$-symbols}, which can be similarly from the recoupling coefficients of $\SU(2)$.%
\footnote{However, for our purposes it was important to use the recoupling coefficients in \cref{main result}, since the dimensions of irreducible $S_k$-representations grow exponentially with $k$ and thus affect the asymptotics.}
It is immediately apparent from the graphical expression that it has the symmetries of a tetrahedron.
We record the following consequence of this symmetry, which has a well-known counterpart for $\SU(2)$; cf.\ \cite[(B4)]{LevinWen05}:

\begin{cor}[Tetrahedral symmetries]
  \label{cor:symmetry}
  The quantities
  \begin{equation}
    \frac 1 {d_\mu d_\nu} \norm{\sixj \alpha \beta \gamma \mu \nu \lambda}_{\text{HS}}^2
  \end{equation}
  are invariant under exchanging the columns $(\beta,\lambda) \leftrightarrow (\mu,\nu)$ and also under exchanging the columns $(\alpha,\gamma) \leftrightarrow (\mu,\nu)$.
\end{cor}
\begin{proof}
  This is an immediate consequence of \cref{prp:invariant object}, since the right-hand side norm in \eqref{eq:invariant norm} is invariant under reflection of the diagram by the axes through the edges labeled by $\alpha$ and $\beta$, respectively.
\end{proof}

\subsection*{Algebraic proof}

We now give an alternative, algebraic proof of \cref{prp:invariant object} and \cref{cor:symmetry} that follows along the same lines as the graphical proof.
In quantum information theory, \emph{maximally entangled states} on a Hilbert space $\calH \otimes \calH$ are defined by the formula
\begin{equation}
\label{eq:epr}
  \ket{\Psi^+_\calH} := \frac 1 {\sqrt{\dim \mathcal H}} \sum_i \ket i \otimes \ket i
\end{equation}
with respect to an orthonormal basis $\ket i$.
They satisfy the fundamental identity
\begin{equation}
  \label{eq:epr transpose}
  (X \otimes \id) \ket{\Psi^+_\calH} = (\id \otimes X^T) \ket{\Psi^+_\calH}
\end{equation}
for any operator $X$ on $\calH$, where $X^T$ denotes the transpose in the basis $\ket i$.
Thus they are invariant under operations of the form $U \otimes \overline{U}$, where $U \in \U(\calH)$ is a unitary and where $\overline{U}$ denotes its complex conjugate with respect to the basis $\ket i$ \cite{HorodeckiHorodecki99}.
In particular, this implies that for any basis $\ket{\lambda,i}$ of $[\lambda]$ in which $S_k$ acts by orthogonal transformations,
\begin{equation*}
  \ket{\Psi^+_\lambda} := \frac 1 {\sqrt{\dim [\lambda]}} \sum_j \ket{\lambda,j} \otimes \ket{\lambda,j}
\end{equation*}
is the (unique up to phase) invariant vector in $[\lambda] \otimes [\lambda]$---as we had asserted before above \cref{Special CG}.
By using \eqref{eq:epr transpose} it is straightforward to verify that the following two well-known properties hold:
\begin{itemize}
\item We have the ``teleportation identity''
  \begin{equation}
  \label{eq:bell teleport}
    \left( \bra{\Psi^+_\calH} \otimes \id_{\mathcal H} \right)
   \left( \id_{\mathcal H} \otimes \ket {\Psi^+_\calH} \right)
   = \frac 1 {\dim \mathcal H} \id_{\mathcal H}.
  \end{equation}
  This is the algebraic version of \cref{Two bubbles}.
  It follows that for any two operators $X \colon \calK \rightarrow \calK' \otimes \calH$ and $Y \colon \calH \otimes \calL \rightarrow \calL'$ we have the relation
  \begin{equation}
  \label{eq:bell teleport operators}
  \begin{aligned}
      &\left( \id_{\calK'} \otimes \bra{\Psi^+_\calH} \otimes \id_{\calL'} \right)
      \left( X \otimes \id_\calH \otimes Y \right)
      \left( \id_\calK \otimes \ket{\Psi^+_\calH} \otimes \id_\calL \right) \\
    = &\frac 1 {\dim \calH} (\id_{\calK'} \otimes Y) (X \otimes \id_\calL)
  \end{aligned}
  \end{equation}
  (which is best understood graphically).
\item The normalized trace of any operator $X$ can be written as
  \begin{equation}
  \label{eq:bell trace}
    \braket{\Psi^+_\calH | X \otimes \id_{\mathcal H} | \Psi^+_\calH} =
    \frac 1 {\dim \mathcal H} \tr X.
  \end{equation}
\end{itemize}
For any $\alpha$, $\beta$ and $\lambda$, we shall consider the following sets of vectors in $([\alpha] \otimes [\beta] \otimes [\lambda])^{S_k}$,
\begin{align}
  \ket{\alpha\beta\lambda,i} &:= (\Phi^{\alpha\beta}_{\lambda,i} \otimes \id_\lambda) \ket{\Psi^+_\lambda}, \nonumber \\
  \ket{\widetilde{\alpha\beta\lambda,i}} &:= \sqrt {\frac {d_\alpha d_\beta} {d_\lambda}} (\Phi^{\alpha\beta}_{\lambda,i} \otimes \id_\alpha \otimes \id_\beta)^\dagger (\ket{\Psi^+_\alpha} \otimes \ket{\Psi^+_\beta}), \label{eq:alternative basis}
\end{align}
constructed as in \cref{lem:invariant bases}.
We now prove algebraically that each set forms an orthonormal basis.
For the first,
\begin{align*}
    &\braket{\alpha\beta\lambda,i | \alpha\beta\lambda,i'}
  = \braket{\Psi^+_\lambda | (\Phi^{\alpha\beta}_{\lambda,i})^\dagger \Phi^{\alpha\beta}_{\lambda,i'} \otimes \id_\lambda | \Psi^+_\lambda}
  = \frac 1 {d_\lambda} \tr (\Phi^{\alpha\beta}_{\lambda,i})^\dagger \Phi^{\alpha\beta}_{\lambda,i'} \\
  = &\braket{\Phi^{\alpha\beta}_{\lambda,i}, \Phi^{\alpha\beta}_{\lambda,i'}}_\lambda
  = \delta_{i,i'}
\end{align*}
by \eqref{eq:bell trace} and the definition of the inner product.
For the second set of vectors,
\begin{align*}
    &\braket{\widetilde{\alpha\beta\lambda,i} | \widetilde{\alpha\beta\lambda,i'}}
  = \frac {d_\alpha d_\beta} {d_\lambda}
    \braket{\Psi^+_\alpha \otimes \Psi^+_\beta |
      \Phi^{\alpha\beta}_{\lambda,i} (\Phi^{\alpha\beta}_{\lambda,i'})^\dagger \otimes \id_\alpha \otimes \id_\beta |
      \Psi^+_\alpha \otimes \Psi^+_\beta} \\
  = &\frac 1 {d_\lambda} \tr \Phi^{\alpha\beta}_{\lambda,i} (\Phi^{\alpha\beta}_{\lambda,i'})^\dagger
  = \braket{\Phi^{\alpha\beta}_{\lambda,i'}, \Phi^{\alpha\beta}_{\lambda,i}}_\lambda
  = \delta_{i',i}.
\end{align*}

We now consider the recoupling coefficients. First, \eqref{6j via clebsch gordan} and \eqref{eq:bell trace} give
\begin{align*}
    &\sixj\alpha \beta\gamma\mu\nu\lambda_{ij}^{kl}
  = \frac 1 {d_\lambda} \tr
    (\Phi^{\alpha\nu}_{\lambda,k})^\dagger \left( \id_\alpha \otimes \Phi^{\beta\gamma}_{\nu,l} \right)^\dagger
    \left( \Phi^{\alpha\beta}_{\mu,j} \otimes \id_{\gamma} \right) \Phi^{\mu\gamma}_{\lambda,i} \\
  = &\bra{\Psi^+_\lambda}
     \left( (\Phi^{\alpha\nu}_{\lambda,k})^\dagger \otimes \id_\lambda \right)
     \left( \id_\alpha \otimes \Phi^{\beta\gamma}_{\nu,l} \otimes \id_\lambda \right)^\dagger
     \left( \Phi^{\alpha\beta}_{\mu,j} \otimes \id_\gamma \otimes \id_\lambda \right)
     \underbrace{\left( \Phi^{\mu\gamma}_{\lambda,i} \otimes \id_\lambda \right)
     \ket{\Psi^+_\lambda}}_{=\ket{\mu\gamma\lambda,i}}.
\end{align*}
We may now apply \eqref{eq:bell teleport operators} to $X = \ket{\mu\gamma\lambda,i}$ and $Y = \Phi^{\alpha\beta}_{\mu,j}$ in order to rewrite 
\begin{align*}
    &\left( \Phi^{\alpha\beta}_{\mu,j} \otimes \id_\gamma \otimes \id_\lambda \right) \ket{\mu\gamma\lambda,i} \\
  = &d_\mu \left( \id_\gamma \otimes \id_\lambda \otimes \bra{\Psi^+_\mu} \otimes \id_\alpha \otimes \id_\beta \right)
    \left( \ket{\mu\gamma\lambda,i} \otimes \id_\mu \otimes \Phi^{\alpha\beta}_{\mu,j} \right)
    \ket{\Psi^+_\mu} \\
  = &d_\mu \left( \id_\gamma \otimes \id_\lambda \otimes \bra{\Psi^+_\mu} \otimes \id_\alpha \otimes \id_\beta \right)
    \left( \ket{\mu\gamma\lambda,i} \otimes \ket{\alpha\beta\mu,j} \right).
\end{align*}
Continuing in this way and using definition \eqref{eq:alternative basis}, we obtain the following expression for the matrix elements of the recoupling coefficient:
\begin{align*}
  \sixj\alpha \beta\gamma\mu\nu\lambda_{ij}^{kl}
  = &d_\mu d_\alpha d_\beta d_\gamma d_\nu \sqrt{\frac {d_\lambda} {d_\alpha d_\nu}} \sqrt{\frac {d_\nu} {d_\beta d_\gamma}} \\
  \times &\left( \bra{\Psi^+_\alpha} \otimes \bra{\Psi^+_\beta} \otimes \bra{\Psi^+_\gamma} \otimes \bra{\Psi^+_\mu} \otimes \bra{\Psi^+_\nu} \otimes \bra{\Psi^+_\lambda} \right) \\
  &\left( \ket{\mu\gamma\lambda,i} \otimes \ket{\alpha\beta\mu,j} \otimes \ket{\widetilde{\alpha\nu\lambda,k}} \otimes \ket{\widetilde{\beta\gamma\nu,l}} \right)
\end{align*}
The sum of their absolute values squared over all indices $i$, $j$, $k$ and $l$ is equal to
\begin{equation}
\label{eq:algebraic recoupling norm squared}
\begin{aligned}
    \frac 1 {d_\mu d_\nu} \normHS{\sixj\alpha\beta\gamma\mu\nu\lambda}^2
  =\ &d_\alpha d_\beta d_\gamma d_\lambda d_\mu d_\nu \\
  \times &\tr \left( P^{\alpha\alpha} \otimes P^{\beta\beta} \otimes P^{\gamma\gamma} \otimes P^{\mu\mu} \otimes P^{\nu\nu} \otimes P^{\lambda\lambda} \right) \\
  &\quad \left( P^{\mu\gamma\lambda} \otimes P^{\alpha\beta\mu} \otimes P^{\alpha\nu\lambda} \otimes P^{\beta\gamma\nu} \right),
\end{aligned}
\end{equation}where $P^{\alpha\beta\mu}$ denotes the orthogonal projection onto $([\alpha] \otimes [\beta] \otimes [\mu])^{S_k}$,
$P^{\alpha\alpha} = \proj{\Psi^+_\alpha}$, etc.
Equation~\eqref{eq:algebraic recoupling norm squared} is the algebraic analogue of \cref{prp:invariant object}.
As before, \cref{cor:symmetry} is a direct consequence of its symmetries.

\section{Entropy inequalities from symmetries: strong subadditivity}
\label{sec:ssa}

We now prove the strong subadditivity and weak monotonicity of the von Neumann entropy as a direct consequence of \cref{main result} and the symmetry properties in \cref{cor:symmetry}.

To start, we note that it follows from the first invariance asserted in \cref{cor:symmetry} and the polynomial upper bound~\eqref{6j lower upper} that
\begin{equation}
  \label{hilbert schmidt ratio}
  \normHS{\sixj \alpha \beta \gamma \mu \nu \lambda}^2
  = \dfrac{d_\mu d_\nu} {d_\beta d_\lambda}
  \normHS{\sixj \alpha \mu \gamma \beta \lambda \nu}^2
  \leq \poly(k) \dfrac{d_\mu d_\nu} {d_\beta d_\lambda}.
\end{equation}
Hence if $\rho_{ABC}$ is a density operator on $\CC^{abc}$ then \cref{main result} implies that
\begin{equation}
  \label{subadd ratio cor}
  \dfrac{d_\mu d_\nu} {d_\beta d_\lambda}
  \geq
  \frac 1 {\mathrm{poly}(k)}
\end{equation}
for sequences of normalized Young diagrams that converge to the respective spectra of the reduced density operators.
Since for large $k$,  $\frac 1 {k} \log_2 \dim [\lambda] \rightarrow H(\bar\lambda) = \sum_i -\bar\lambda_i \log_2 \bar\lambda_i$ \cite{christandlmitchison06}, we conclude that the von Neumann entropy is strongly subadditive:

\begin{thm}[Strong subadditivity of von Neumann entropy~\cite{liebruskai73}]
  For any density operator $\rho_{ABC}$ on $\CC^{abc}$,
  \begin{equation}
  \label{eq:ssa}
  	S(\rho_{AB}) + S(\rho_{BC}) \geq S(\rho_{B}) + S(\rho_{ABC}),
  \end{equation}
  where $S(\rho) := -\tr \rho \log \rho$ denotes the von Neumann entropy of a density operator $\rho$.
\end{thm}


For $[\beta]$ the trivial representation, this proof of strong subadditivity reduces to the proof of subadditivity given in \cite{christandlmitchison06} (cf.\ the discussion at the end of \cref{sec:tripartite}).
The weak monotonicity,
\[
  S(\rho_{AB}) + S(\rho_{BC}) \geq S(\rho_A) + S(\rho_C),
\]
follows similarly from the second invariance in \cref{cor:symmetry}.

As we have mentioned in the introduction, it is an important open question to decide whether the von Neumann entropy satisfies any other linear entropy inequalities beyond strong subadditivity and weak monotonicity.
Our proofs of the latter are markedly different from previous proofs in the literature, which are built on operator convexity \cite{liebruskai73,NielsenPetz05,Ruskai07a,Effros09} or asymptotic equipartition \cite{Renner05,Gromov13} (cf.\ the review \cite{Ruskai05}).
In our approach, we interpret an entropy inequality as the asymptotic shadow of a dimensional relation such as~\eqref{subadd ratio cor}.
We establish the latter by exploiting the symmetries of a corresponding representation-theoretic object -- the recoupling coefficients -- together with a lower bound from spectrum estimation.
This hints towards an intriguing route towards establishing new entropy inequalities---namely, by constructing novel representation-theoretic objects (e.g., by composing Clebsch--Gordan maps) and uncovering their symmetries (as can conveniently be done using the graphical calculus).

\section{Sums of matrices and quantum marginals}
We have earlier discussed Wigner's characterization of the asymptotics of recoupling coefficients for $\SU(2)$ in terms of the existence of Euclidean tetrahedra.
In this section, we will see that his result can in fact be related to \cref{main result} on a precise mathematical level.
Before we describe the construction, we note that the existence of a tetrahedron with side lengths $j_A$, $j_B$, $j_C$, $j_{AB}$, $j_{BC}$, $j_{ABC}$ is equivalent to the existence of vectors $\vec j_A$, $\vec j_B$, $\vec j_C$ such that $\abs{\vec j_A} = j_A$, $\abs{\vec j_B} = j_B$, $\abs{\vec j_C} = j_C$, $\abs{\vec j_A + \vec j_B} = j_{AB}$, $\abs{\vec j_B + \vec j_C} = j_{BC}$, and $\abs{\vec j_A + \vec j_B + \vec j_C} = j_{ABC}$.
By the usual identification of $\RR^3$ with the Lie algebra $\mathfrak{su}_2$, we may assign to each vector $\vec j \in \RR^3$ the Hermitian $2 \times 2$-matrix $\vec j \cdot \vec \sigma$, where $\vec\sigma = (\sigma_x, \sigma_y, \sigma_z)$ is the vector of Pauli matrices.
Then the above becomes an instance of the following general problem \cite{backens10}: 
\begin{prb}[Partial sums of matrices]
\label{prb:generalized weyl}
	Do there exist Hermitian $d \times d$-matrices $\calA$, $\calB$ and $\calC$ with given prescribed eigenvalues for $\calA$, $\calB$, $\calC$, $\calA + \calB$, $\calB + \calC$ and $\calA + \calB + \calC$?
\end{prb}
This is a natural generalization of the problem of determining the relation between the eigenvalues of $A$, $B$ and $A+B$, posed by Weyl, whose solution conjectured by Horn \cite{horn62} was proved in the celebrated works \cite{Klyachko-Horn,KnuTao99}.

In \cite{klyachko04}, it was shown how the one-body quantum marginal problem degenerates to Weyl's problem in an appropriate limit.
We will now show that \cref{prb:generalized weyl} can likewise be considered as a special case of the quantum marginals problem for overlapping subsystems characterized by \cref{main result}---both on the level of geometry and on the level of representation theory.

\medskip

Let $\calA$, $\calB$, $\calC$ be Hermitian $d \times d$-matrices.
Without loss of generality, we may assume that $\calA, \calB, \calC \geq 0$ and that $1 - \tr (\calA + \calB + \calC) \geq \norm{\calA+\calB+\calC}_\infty$ (else, we may add suitable multiples of the identity and rescale).
Generalizing a construction from \cite{christandl-horn}, we define a tripartite density operator $\rho_{ABC}$ as the reduced density operator of the four-party pure state
\begin{equation}
\label{eq:purification embedding}
\begin{aligned}
  \ket{\psi}_{ABCD}
&=\;\sum_{i=1}^d \sqrt{\calA} \ket i_A \otimes \ket{00i}_{BCD}
+ \sum_{j=1}^d \sqrt{\calB} \ket j_B \otimes \ket{00j}_{ACD} \\
+\;&\sum_{k=1}^d \sqrt{\calC} \ket k_C \otimes \ket{00k}_{ABD}
+ \sqrt{1 - \tr (\calA + \calB + \calC)} \ket{0000}_{ABCD}
\end{aligned}
\end{equation}
in $(\CC^{d+1})^{\otimes 4}$, where we consider $\calA$, $\calB$ and $\calC$ as acting on $\CC^d$, and $\CC^{d+1} = \CC \ket 0 \oplus \CC^d$.
We record the following properties:

\begin{lem}
\label{lem:purification embedding}
  Let $\rho_{ABC}$ be the quantum state with purification \eqref{eq:purification embedding}.
  Then the non-zero eigenvalues of $\rho_{ABC}$ and \emph{all} its reduced density operators are given by
  \begin{align*}
    \spec \rho_{ABC} &= (1 - \tr (\calA+\calB+\calC), \spec (\calA+\calB+\calC)), \\
    \spec \rho_{AB} &= (1 - \tr (\calA+\calB), \spec (\calA+\calB)), \\
    \spec \rho_{A} &= (1 - \tr \calA, \spec \calA),
  \end{align*}
  etc.
\end{lem}
\begin{proof}
  Observe that $\ket{\psi_{ABCD}}$ is built from a sum of (unnormalized) maximally entangled states \eqref{eq:epr} on $AD$, $BD$ and $CD$, respectively.
  By using \eqref{eq:epr transpose} and the orthogonality properties of the construction \eqref{eq:purification embedding}, we thus find that
  \begin{equation*}
    \rho_D = \left( 1-\tr(\calA+\calB+\calC) \right) \proj 0 + \left( \overline{\calA+\calB+\calC} \right).
  \end{equation*}
  Thus the density operator $\rho_D$ is block-diagonal with respect to $\CC \proj 0 \oplus \CC^d$.
  These blocks can be jointly diagonalized, and our assumption $1 - \tr(\calA+\calB+\calC) \geq \norm{\calA+\calB+\calC}_\infty$ implies that $1 - \tr(\calA+\calB+\calC)$ is the largest eigenvalue.
  Since moreover the spectrum of a Hermitian matrix is invariant under conjugation, this shows the first claim, as $\rho_{ABC}$ and $\rho_D$ have the same non-zero eigenvalues.

  If we only trace out the first two systems, then we instead get a block decomposition of the form
  \begin{equation*}
    \rho_{CD} = \proj 0_C \otimes \left( \overline{\calA+\calB} \right) + \proj\phi_{CD},
  \end{equation*}
  where $\ket\phi_{CD} = \sum_{k=1}^d \sqrt{\calC} \ket k_C \otimes \ket k_D + \sqrt{1 - \tr (\calA + \calB + \calC)} \ket{00}_{CD}$.
  Using \eqref{eq:bell trace}, we find that $\braket{\phi_{CD} | \phi_{CD}} =  1 - \tr(\calA+\calB)$, so that the second claim follows as above.

  The last claim follows from
  \begin{align*}
    \rho_A
    = &\sum_i \sqrt{\calA} \proj i_A \sqrt{\calA}
    \;+\; \Bigl( \tr \calB + \tr \calC + \bigl( 1 - \tr (\calA + \calB + \calC) \bigr) \Bigr) \proj 0_A \\
    = &\calA + \left( 1 - \tr \calA \right) \proj 0_A,
  \end{align*}
  which is established similarly.
  All other marginal spectra can be computed in the same way.
\end{proof}

We have thus obtained an embedding of triples of matrices into the space of tripartite density operators that preserves the eigenvalue information.
We remark that \cref{lem:purification embedding} can in particular be used to obtain entropy inequalities for convex combinations of Hermitian matrices from entropy inequalities for multipartite quantum states. E.g., strong subadditivity~\eqref{eq:ssa} yields after some manipulations the following inequality
\begin{align*}
&\quad h(a+b) + (a+b) S\Bigl(\frac{\calA+\calB}{a+b}\Bigr) + h(b+c) + (b+c) S\Bigl(\frac{\calB+\calC}{b+c}\Bigr) \\
&\geq h(b) + b S\Bigl(\frac \calB b\Bigr) + h(a+b+c) + (a+b+c) S\Bigl(\frac{\calA+\calB+\calC}{a+b+c}\Bigr),
\end{align*}
where we have abbreviated $a = \tr\calA$, $b = \tr\calB$, and $c = \tr\calC$, and where $h(x) := -x \log x - (1-x) \log (1-x)$ denotes the binary entropy function.

\begin{cor}
\label{cor:polygonal}
  The state $\rho_{ABC}$ has rank at most $d+1$ and it satisfies the equality 
  \begin{equation}
  \label{eq:polygonal equality}
    r_{A,1} + r_{B,1} + r_{C,1} = 2 + r_{ABC,1},
  \end{equation}
  where $r_{I,1}$ denotes the maximal eigenvalue of the reduced density operator $\rho_I$.
\end{cor}

In other words, the states $\rho_{ABC}$ constructed above saturate the ``polygonal inequality''~\cite{higuchi_sudbery_szulc_2003_onequbit}
\begin{equation}
\label{polygonal ieq}
	r_{A,1} + r_{B,1} + r_{C,1} \leq 2 + r_{ABC,1},
\end{equation}
which holds for arbitrary quantum states, with equality.
We now show the following converse to \cref{cor:polygonal}.

\begin{prp}
\label{prp:embedding converse}
  Let $\rho_{ABC}$ be a tripartite quantum state on $(\CC^{d+1})^{\otimes 3}$ of rank at most $d+1$ that satisfies the equality \eqref{eq:polygonal equality}. Then $\rho_{ABC}$ can up to local unitaries be purified in the form \eqref{eq:purification embedding}.
\end{prp}
\begin{proof}
  For this, we recall the following proof of the inequality~\eqref{polygonal ieq}.
  Let $\ket 0_A$, $\ket 0_B$, and $\ket 0_C$ denote maximal eigenvectors of $\rho_A$, $\rho_B$ and $\rho_C$, respectively.
  Let $P_A$, $P_B$ and $P_C$ denote the corresponding orthogonal projectors, and set $P_A^\perp := \id - P_A$, etc.
  Then,
  \begin{align*}
      &r_{A,1} + r_{B,1} + r_{C,1} \\
    =\;&\tr \rho_{ABC} (P_A \otimes \id_{BC} + P_B \otimes \id_{AC} + P_C \otimes \id_{AB}) \\
    =\;&2 \tr \rho_{ABC} \id_{ABC}
    + \tr \rho_{ABC} (P_A \otimes P_B \otimes P_C) \\
    -\;&2 \tr \rho_{ABC} (P_A^\perp \otimes P_B^\perp \otimes P_C^\perp)
    - \tr \rho_{ABC} (P_A \otimes P_B^\perp \otimes P_C^\perp) \\
    -\;&\tr \rho_{ABC} (P_A^\perp \otimes P_B \otimes P_C^\perp)
    - \tr \rho_{ABC} (P_A^\perp \otimes P_B^\perp \otimes P_C) \\
    \leq\;&2 + \braket{000 | \rho_{ABC} | 000} \leq 2 + r_{ABC,1}.
  \end{align*}
  The first inequality is obtained by omitting the terms with negative signs, and the second by using the variational principle for the maximal eigenvalue of $\rho_{ABC}$.
  It is thus immediate that we have equality if and only if $\ket{000}_{ABC}$ is a maximal eigenvector of $\rho_{ABC}$ and
  \begin{equation}
  \label{eq:negative sign projectors}
  \begin{aligned}
    &\tr \rho_{ABC} (P_A^\perp \otimes P_B^\perp \otimes P_C^\perp) =
    \tr \rho_{ABC} (P_A^\perp \otimes P_B^\perp \otimes P_C) \\
    =\;&\tr \rho_{ABC} (P_A^\perp \otimes P_B \otimes P_C^\perp) =
    \tr \rho_{ABC} (P_A \otimes P_B^\perp \otimes P_C^\perp) = 0.
  \end{aligned}
  \end{equation}
  Let us now assume that this is the case. Since the rank of $\rho_{ABC}$ was assumed to be at most $d+1$, we can find a purification on $(\CC^{d+1})^{\otimes 4}$.
  Since $\ket{000}_{ABC}$ is a maximal eigenvector, we can arrange for the first term of the Schmidt decomposition to be $\sqrt{r_{ABC,1}} \ket{000}_A \otimes \ket{0}_D$. In other words, the purification can be chosen of the form
  \begin{equation*}
    \ket{\psi_{ABCD}} = \sqrt{r_{ABC,1}} \ket{0000}_{ABCD} + \sum_{i,j,k,l} \psi_{ijkl} \ket{ijkl}_{ABCD}.
  \end{equation*}
  A priori, the right-hand side can run over all indices $(i,j,k) \neq (0,0,0)$ and $l \neq 0$ by orthogonality of the bases in the Schmidt decomposition.
  But \eqref{eq:negative sign projectors} implies that in fact precisely two out of the three indices $(i,j,k)$ have to be zero, so that we obtain
  \begin{align*}
    \ket{\psi_{ABCD}} = \sqrt{r_{ABC,1}} \ket{0000}_{ABCD}
    + \sum_{i,l=1}^d \psi_{i00l} \ket{i00l}_{ABCD} \\
    + \sum_{j,l=1}^d \psi_{0j0l} \ket{0j0l}_{ABCD}
    + \sum_{k,l=1}^d \psi_{00kl} \ket{00kl}_{ABCD}.
  \end{align*}
  Thus we may define $d \times d$-matrices $X_A$, $X_B$ and $X_C$ such that
  \begin{align*}
    \ket{\psi_{ABCD}} = \sqrt{r_{ABC,1}} \ket{0000}_{ABCD}
    + \sum_i X_A \ket i_A \otimes \ket{00i}_{BCD} \\
    + \sum_j X_B \ket j_B \otimes \ket{00j}_{ACD}
    + \sum_k X_C \ket k_C \otimes \ket{00k}_{ABD}.
  \end{align*}
  Finally, we use the polar decomposition to write $X_A = U_A \abs{X_A}$, etc., and set $\sqrt{\calA} := \abs{X_A}$, etc. Then \eqref{eq:purification embedding} is indeed a purification of the quantum state $(U^\dagger_A \otimes U^\dagger_B \otimes U^\dagger_C) \rho_{ABC} (U_A \otimes U_B \otimes U_C)$, which is locally unitarily equivalent to $\rho_{ABC}$.
\end{proof}

The following theorem shows that \cref{prb:generalized weyl} -- and, in particular, the existence of Wigner's tetrahedra -- is in a precise mathematical sense a special case of the quantum marginal problem with overlapping marginals covered by \cref{main result}.
This generalizes the corresponding result for the one-body quantum marginal problem mentioned above and in particular gives a geometric proof of the latter.

\begin{thm}
\label{thm:weyl wigner embedding}
  Let $s_{\calA}$, $s_{\calB}$, $s_{\calC}$, $s_{\calA+\calB}$, $s_{\calB+\calC}$, $s_{\calA+\calB+\calC}$ be vectors in $\RR^d_{\geq 0}$ with weakly decreasing entries.
  Assume that $\norm{s_{\calA}}_1 + \norm{s_{\calB}}_1 = \norm{s_{\calA+\calB}}_1$, etc., and that $1 - \norm{s_{\calA+\calB+\calC}}_1 \geq s_{\calA+\calB+\calC,1}$.%
\footnote{The former requirements correspond to the equalities $\tr \calA + \tr \calB = \tr (\calA + \calB)$, etc., and therefore are clearly necessary. The latter can always be achieved by rescaling.}
  Then the following are equivalent:
  \begin{enumerate}
    \item[(1)] There exist Hermitian $d \times d$-matrices $\calA$, $\calB$ and $\calC$ with
    $\spec (\calA+\calB+\calC) = s_{\calA+\calB+\calC}$,
    $\spec (\calA+\calB) = s_{\calA+\calB}$,
    $\spec \calA = s_A$,
    etc.\ as their partial sums.
    \item[(2)] There exists a quantum state $\rho_{ABC}$ on $(\CC^{d+1})^{\otimes 3}$ with non-zero eigenvalues
    $\spec \rho_{ABC} = (1-\norm{s_{\calA+\calB+\calC}}_1, s_{\calA+\calB+\calC})$,
    $\spec \rho_{AB} = (1-\norm{s_{\calA+\calB}}_1,$ $s_{\calA+\calB})$,
    $\spec \rho_{A} = (1-\norm{s_{\calA}}_1, s_{\calA})$,
    etc.\ for their reduced density operators.
  \end{enumerate}
\end{thm}
\begin{proof}
  $(1) \Rightarrow (2)$ is the content of \cref{lem:purification embedding}.
  For $(2) \Rightarrow (1)$, we use \cref{prp:embedding converse} to obtain Hermitian $d \times d$-matrices $\calA$, $\calB$, $\calC$ such that $\rho_{ABC}$ is locally unitarily equivalent to the state $\rho'_{ABC}$ with purification \eqref{eq:purification embedding}.
  Since the spectra of $\rho_{ABC}$ and its reduced density operators are left invariant by local unitaries, \cref{lem:purification embedding} implies that the partial sums of these matrices $\calA$, $\calB$ and $\calC$ have the desired spectra.
\end{proof}

Similar statements can be proved for all marginal spectra (i.e., including $s_{\calA+\calC}$, since \cref{lem:purification embedding} holds for all reduced density operators) as well as for an arbitrary number of summands. Thus the quantum marginal problem with overlaps is a precise generalization of the problem of characterizing the eigenvalues of partial sums of Hermitian matrices.

\subsection*{Recoupling Coefficients of the Unitary and Symmetric Groups}%

We now show an analogous statement to \cref{thm:weyl wigner embedding} on the level of representation theory---namely, that the recoupling coefficients of the unitary group can be obtained as special recoupling coefficients of the symmetric group.

To see this, let $\lambda$ be a Young diagram.
In \cite{Nishiyama00}, the restriction of an irreducible $\U(k)$-representation $V^k_\lambda$ to the \emph{subgroup} of permutation matrices $S_k \subseteq U(k)$ has been computed:
\begin{equation}
\label{eq:nishiyama}
  V^k_\lambda \big|_{S_k \subseteq \U(k)}
  = \bigoplus_{\mu \vdash_k \abs\lambda}
    \Ind_{S_\alpha \times S_{k - \abs\alpha}}^{S_k}
    \left( [\lambda]^{S_\mu} \otimes \mathbf 1 \right)
\end{equation}
In the right-hand side of \eqref{eq:nishiyama}, $\Ind$ denotes an \emph{induced representation} and $\alpha$ is the unique Young diagram with number of boxes $\abs\alpha$ equal to the number of rows of $\mu$ such that
\[N_{S_{\abs\mu}}(S_\mu) / S_\mu \cong S_\alpha,\]
where $S_\mu := S_{\mu_1} \times S_{\mu_2} \times \dots \subseteq S_{\abs\mu}$ is the Young subgroup corresponding to $\mu$,
$N_{S_{\abs\mu}}(S_\mu)$ its normalizer in $S_{\abs\mu}$,
and $S_\alpha \subseteq S_{\abs\alpha}$ the Young subgroup of $\alpha$.
Note that $S_\alpha$ indeed acts on the subspace $[\lambda]^{S_\mu}$.

\begin{lem}
  If $k - \abs\lambda \geq \lambda_1$, then $\lambda' := (k-\abs\lambda,\lambda)$ is again a Young diagram, and
  \begin{equation}
  \label{eq:nishiyama lemma}
    V^k_\lambda \big|_{S_k \subseteq \U(k)} \cong [\lambda'] \oplus \dots.
  \end{equation}
  Here and in the following, we write ``\dots'' for a sum of irreducible $S_k$-representations whose Young diagrams have \emph{longer} first rows than all the preceding ones.

  Otherwise, if $k - \abs\lambda < \lambda_1$ then the first row of any Young diagram that appears in the restriction of $V^k_\lambda$ is longer than $k - \abs\lambda$. 
\end{lem}
\begin{proof}
  Since induction is transitive, we can rewrite \eqref{eq:nishiyama} as
  \begin{equation}
  \label{eq:nishiyama double induction}
  V^k_\lambda \big|_{S_k \subseteq \U(k)}
  = \bigoplus_{\mu \vdash_k \abs\lambda}
    \Ind_{S_{\abs\alpha} \times S_{k - \abs\alpha}}^{S_k}
    \left( \Ind_{S_\alpha}^{S_{\abs\alpha}} \left( [\lambda]^{S_\mu} \right) \otimes \mathbf 1 \right).
  \end{equation}
  The \emph{Pieri formula} asserts that the $S_k$-representation induced from a tensor product of an irreducible $S_{\abs\alpha}$-representation $[\nu]$ with the trivial $S_{k-\abs\alpha}$-representation $\mathbf 1$ is given by the sum over all irreducible $S_k$-representations with a Young diagram that can be obtained by adding $k-\abs\alpha$ boxes to $\nu$, with no two in the same column (see, e.g., \cite[\S{}2.2, (4)]{Fulton97}).
  The first row of any such Young diagram is of length at least $k-\abs\alpha$.
  As $\abs\alpha$ is equal to the number of rows of $\mu$, we obtain the lower bound
  \begin{equation}
  \label{eq:nishiyama lower bound}
    k-\abs\alpha \geq k-\abs\mu = k-\abs\lambda
  \end{equation}
  on the length of the first row of any irreducible $S_k$-representation that occurs in the restriction of $V^k_\lambda$.

  Equality in \eqref{eq:nishiyama lower bound} can occur only if each row of $\mu$ contains a single box, i.e., for $\mu=(1,\ldots,1,0,\ldots,0)$, such that $\abs\alpha = \abs\mu = \abs\lambda$.
  Then $S_\mu$ is the trivial group, $S_\alpha = S_{\abs\alpha} = S_{\abs\lambda}$, and the corresponding summand in \eqref{eq:nishiyama double induction} is equal to
  \begin{equation}
  \label{eq:nishiyama only way}
    \Ind_{S_{\abs\lambda} \times S_{k - \abs\lambda}}^{S_k} \left( [\lambda] \otimes \mathbf 1 \right).
  \end{equation}
  By the Pieri formula, \eqref{eq:nishiyama only way} contains an irreducible $S_k$-representation with first row of length $k-\abs\lambda$ if and only if $\lambda_1 \leq k-\abs\lambda$ (since we only add boxes to $\lambda$).
  Moreover, if this condition is satisfied then there is only a single option, namely to place one box in each of the $k-\abs\lambda$ leftmost columns, resulting in the Young diagram $\lambda' = (k-\abs\lambda,\lambda)$.
\end{proof}

We now consider the decomposition of a tensor product of irreducible $\U(k)$-representations,
\begin{equation*}
  V^k_\alpha \otimes V^k_\beta = \bigoplus_\lambda c^{\alpha,\beta}_\lambda V^k_\lambda,
\end{equation*}
where we assume that $k - \abs\alpha \geq \alpha_1$ and $k - \abs\beta \geq \beta_1$.
The multiplicities $c^{\alpha,\beta}_\lambda$ are known as the \emph{Littlewood--Richardson coefficients}, and they are independent of the choice of $k$ (if $k$ is at least as large as the number of rows in the Young diagrams involved) \cite{JamesKerber81}.
Moreover, $c^{\alpha,\beta}_\lambda$ is non-zero only if $\abs\alpha + \abs\beta = \abs\lambda$.
It follows from the points above that
\begin{equation*}
    V^k_\alpha \otimes V^k_\beta \big|_{S_k \subseteq \U(k)}
  = \bigoplus_{\mathclap{\substack{\abs\alpha + \abs\beta = \abs\lambda}}} c^{\alpha,\beta}_\lambda V^k_\lambda  \big|_{S_k \subseteq \U(k)}
  = \left(\;\;\;\;\;\bigoplus_{\mathclap{\substack{\abs\alpha + \abs\beta = \abs\lambda, \\ k - \abs\lambda \geq \lambda_1}}} c^{\alpha,\beta}_\lambda [\lambda'] \right) \oplus \dots
\end{equation*}
On the other hand, by applying \eqref{eq:nishiyama lemma} to the individual tensor factors we find that
\begin{align*}
  &V^k_\alpha \otimes V^k_\beta \big|_{S_k \subseteq \U(k)}
  = \left( [\alpha'] \oplus \dots \right) \otimes \left( [\beta'] \oplus \dots \right)
  = \left(\;\;\;\;\;\bigoplus_{\mathclap{\substack{\abs\lambda = \abs\alpha + \abs\beta \\ k - \abs\lambda \geq \lambda_1}}} g_{\alpha',\beta',\lambda'} [\lambda'] \right) \oplus \dots,
\end{align*}
where $g_{\alpha',\beta',\lambda'}$ are the Kronecker coefficients. 
In the last inequality, we have used that $g_{\alpha',\beta',\lambda'} > 0$ only if $\abs\lambda \leq \abs\alpha + \abs\beta$ \cite[Theorem 2.9.22]{JamesKerber81}.
By comparing coefficients we find that
$c^{\alpha,\beta}_\lambda = g_{\alpha',\beta',\lambda'}$ for all triples of Young diagrams with $\abs\alpha + \abs\beta = \abs\gamma$ and $k$ large enough.
We thus recover a well-known result due to Littlewood and Murnaghan that states that the Littlewood--Richardson coefficients are a special case of the Kronecker coefficients \cite{Littlewood58, Murnaghan55}.

What is more, the argument shows that the Clebsch--Gordan embeddings $\Phi^{\alpha'\beta'}_{\lambda'}$ for $S_k$ can be obtained by restricting the ones of $\U(k)$.
In view of \eqref{6j via clebsch gordan}, this implies directly that the recoupling coefficients are the same, since they are built solely from the action on the multiplicity spaces.
Again, the recoupling coefficients for $\U(k)$ do not depend on the choice of $k$ (if $k$ is at least as large as the number of rows in the Young diagrams involved).

\section*{Acknowledgements}
We thank M.~Backens, M.~Gromov, D.~Gross, H.~Haggard, F.~Hellmann, W.~Kami\'{n}ski, A.~Knutson, G.~Mitchison, M.B.~Ruskai, L.~Vinet and R.~Werner for valuable discussions.

We acknowledge financial support by
the German Science Foundation (grant CH 843/2-1),
the Swiss National Science Foundation (grants PP00P2-128455, 20CH21-138799 (CHIST-ERA project CQC)),
the Swiss National Center of Competence in Research `Quantum Science and Technology (QSIT)',
the Swiss State Secretariat for Education and Research supporting COST action MP1006,
the European Research Council under the European Union's Seventh Framework Programme (FP/2007-2013)/ERC Grant Agreement no.\ 337603,
the Simons Foundation,
FQXi,
and a Sapere Aude: DFF-Starting Grant.
M.B.~\c{S}ahino\u{g}lu acknowledges support of the Excellence Scholarship and Opportunity Programme of ETH Z\"urich.

\bibliographystyle{unsrt}
\bibliography{strong6j}

\end{document}